\documentclass[11pt]{scrartcl}

\usepackage{tikz}
\usetikzlibrary{calc,shapes}

\usepackage[color=blue!50]{todonotes}

\usepackage[utf8]{inputenc}
\usepackage[T1]{fontenc}

\usepackage{amsmath,amsthm,amssymb}
\usepackage{mathtools}

\usepackage[bottom=1in,top=1in,left=1in,right=1in]{geometry}
\RedeclareSectionCommand[indent=0pt]{subparagraph}

\usepackage{hyperref}
\usepackage{xspace}
\usepackage{thm-restate}

\newtheorem{theorem}{Theorem}
\newtheorem{definition}[theorem]{Definition}
\newtheorem{lemma}[theorem]{Lemma}

\newtheorem{observation}[theorem]{Observation}
\newtheorem*{technique-no-number}{Technique}
\newtheorem{corollary}[theorem]{Corollary}

\newcommand{\RR}{\mathbb{R}}
\newcommand{\NN}{\mathbb{N}}

\DeclareMathOperator{\opt}{opt}

\newcommand{\fixtoomuchspaceerror}{\vspace{-1.9ex}}

\newcommand{\junk}[1]{}

\title{Analysis of Ward's Method\thanks{This research was supported by ERC Starting Grant 306465 (BeyondWorstCase).}}

\author{\setcounter{footnote}{3}Anna Gro{\ss}wendt\thanks{Department of Computer Science, University of Bonn, Bonn, Germany\newline\hspace*{-0.2cm} \url{grosswen@cs.uni-bonn.de},\hspace*{0.15cm}\url{roeglin@cs.uni-bonn.de},\hspace*{0.15cm}\url{melanieschmidt@uni-bonn.de}}\hspace*{0.15cm} \and Heiko R\"oglin\footnotemark[4] \and Melanie Schmidt\footnotemark[4]}
\date{}

\allowdisplaybreaks

\begin{document}
\setcounter{page}{0}
\maketitle
\thispagestyle{empty}
\begin{abstract}
We study \emph{Ward's method} for the hierarchical $k$-means problem. 
This popular greedy heuristic is based on the \emph{complete linkage} paradigm:
Starting with all data points as singleton clusters, it successively merges two clusters to form a clustering with one cluster less. The pair of clusters is chosen to (locally) minimize the $k$-means cost of the clustering in the next step.

Complete linkage algorithms are very popular for hierarchical clustering problems, yet their theoretical properties have been studied relatively little. For the Euclidean $k$-center problem, Ackermann et al.~\cite{ABKS14} show that the $k$-clustering in the hierarchy computed by complete linkage has a worst-case approximation ratio of~$\Theta(\log k)$. If the data lies in~$\mathbb{R}^d$ for constant dimension $d$, the guarantee improves to $\mathcal{O}(1)$~\cite{GR17}, but the $\mathcal{O}$-notation hides a linear dependence on~$d$. Complete linkage for $k$-median or $k$-means has not been analyzed so far.

In this paper, we show that Ward's method computes a $2$-approximation with respect to the $k$-means objective function if the optimal $k$-clustering is well separated. If additionally the optimal clustering also satisfies a balance condition, then Ward's method fully recovers the optimum solution. These results hold in arbitrary dimension. We accompany our positive results with a lower bound of $\Omega((3/2)^d)$ for data sets in~$\RR^d$ that holds if no separation is guaranteed, and with lower bounds when the guaranteed separation is not sufficiently strong. Finally, we show that Ward produces an $\mathcal{O}(1)$-approximative clustering for one-dimensional data sets.
\end{abstract}

\section{Introduction}

Clustering is a fundamental tool in machine learning. As an unsupervised learning method, it provides an easy way to gain insight into the structure of data without the need for expert knowledge to start with. One of the most popular clustering objectives is $k$-means: Given a set $P$ of points in the Euclidean space~$\mathbb{R}^d$, find $k$ centers that minimize the sum of the squared distances of each point in $P$ to its closest center. The objective is also called \emph{sum of squared errors}, since the centers can serve as representatives, and then the sum of the squared distances becomes the squared error of this representation.

Theory has focused on metric objective functions for a long time: Facility location or $k$-median are very well understood, with upper and lower bounds on the best possible approximation guarantee slowly approaching one another. The $k$-means cost function is arguably more popular in practice, yet its theoretical properties were long not the topic of much analysis. In the last decade, considerable efforts have been made to close this gap.

We now know that $k$-means is NP-hard, even in the plane~\cite{MNV09} and also even for two centers~\cite{ADHP09}. The problem is also APX-hard~\cite{ACKS15}, and the currently best approximation algorithm achieves an approximation ratio of 6.357~\cite{ANSW17}. The best lower bound, though, is only 1.0013~\cite{LSW17}. A seminal paper on $k$-means is the derivation of a practical approximation algorithm, $k$-means++, which is as fast as the most popular heuristic for the problem (the local search algorithm due to Lloyd~\cite{L57}), has an upper bound of $\mathcal{O}(\log k)$ on the expected approximation ratio, and has proven to significantly improve the performance on actual data~\cite{AV07}. Due to its simplicity and superior performance, it (or variants of it) can now be seen as the de facto standard initialization for Lloyd's method. 

From a practical point of view, however, there is still one major drawback of using $k$-means++ and Lloyd's method, and this has nothing to do with its approximation ratio or speed. Before using any method that strives to optimize $k$-means, one has to determine the number~$k$ of clusters. If one knows very little about the data at hand, then even this might pose a challenge. Indeed, there are several suggestions how to set $k$, which usually look at the tradeoff between the number of clusters and the cost (which decreases if the number of clusters is increased). For example, the \emph{elbow method} searches for a point where the cost decreases dramatically, arguing that this happens only at the point of the true number of clusters. However, there are many more methods to choose from (see for example the summary in §5 of~\cite{TWH00}). Notice that one usually needs to compute multiple clusterings for different $k$ to use such a method.

However, there is a simpler and popular method available: hierarchical clustering. Instead of computing  clusterings for several different numbers of clusters and comparing them, one computes one clustering tree (a dendrogram), which contains a clustering for every value of $k$. 
For any $k \in [n-1]$, the $k$-clustering in such a tree results from the $(k+1)$-clustering in the same tree by merging two clusters. 
The hierarchical clustering does not only provide an answer for every $k$, it also allows the user to view the data at different levels of granularity. 
A hierarchical clustering is apparently something very desirable, but the question is: Can the solutions  be good for all values of $k$? Do we lose much by forcing the hierarchical structure?

Dasgupta and Long~\cite{DasguptaL05} were the first to give positive and negative answers to this question.
Their analysis evolves around the (metric) $k$-center problem, which is to minimize the maximum radius of any cluster. 
They compare the $k$-center cost on each level of a hierarchical clustering to an optimal clustering with the best possible radius with the same number of clusters and look for the level with the worst factor. It turns out that popular heuristics for hierarchical clustering can be off by a factor of $\log k$ or even $k$ compared to an optimal clustering. Dasgupta and Long also propose a clever adaption of the $2$-approximation for $k$-center due to Gonz\'alez~\cite{G85}, which results in a hierarchical clustering algorithm. For this algorithm, they can guarantee that the solution is an $8$-approximation of the optimum on every level of the hierarchy simultaneously.

In a series of works, Mettu, and Plaxton~\cite{MP03}, Plaxton~\cite{P06} and finally Lin, Nagarajan, Rajaraman, and Williamson~\cite{LNRW10} develop and refine algorithms for the hierarchical $k$-median problem, which can be seen as the metric cousin of the hierarchical $k$-means problem. It consists of minimizing the sum of the distances of every point to its closest center, and is usually studied in metric spaces. The best known approximation guarantee is $20.06$.
However, the quality guarantee vastly deteriorates for $k$-means: An $\mathcal{O}(1)$-approximation for the hierarchical $k$-means problem follows from~\cite{P06,MP03} as well as from \cite{LNRW10}, but the approximation ratios range between $961$ and $3662$.

On the practical side, however, there is a long known greedy algorithm for the hierarchical $k$-means problem, named \emph{Ward's method}~\cite{W63}.
In the fashion of \emph{complete linkage} algorithms, it does the following.
It starts with singleton clusters, one for each data point from the input~$P\subset\RR^d$. Then it performs $|P|-1$ iterations where two clusters in the current clustering are merged (this is called \emph{agglomerative clustering}). In each iteration, it chooses the pair of clusters which results in the cheapest clustering. This is a locally optimal choice only, since the optimal merge in one operation may prove to be a poor choice with respect to a later level of the hierarchy. 

To the best of the authors' knowledge, the worst-case quality of Ward's method has not been studied so far. In particular, it was not known whether the algorithm can be used to compute constant-factor approximations. We answer this question negatively by giving a family of examples with increasing $k$ and $d$ where the approximation factor of Ward is $\Omega((3/2)^d)$. 

To explain the algorithms popularity, we then proceed to study it under different \emph{clusterability} assumptions. 
Clustering problems are usually NP-hard and even APX-hard, yet clustering is routinely solved in practical applications. This discrepancy has led to the hypothesis that data sets are either easy to cluster, or they have little interesting structure to begin with. \lq Well-clusterable data sets are computationally easy to cluster\rq~\cite{B15} and \lq Clustering is difficult only when it does not matter\rq~\cite{DLS12} are two slogans summarizing this idea. Following it, many notions have been developed that strive to capture how well a data set is clusterable. 
One such notion is \emph{center separation}~\cite{Ben-DavidH14}: A data set~$P\subset\RR^d$ is $\delta$-center separated for some number~$k$ of clusters if the distance between any pair of clusters in the target clustering is at least $\delta$ times the maximal radius of one of the clusters. It satisfies the similar \emph{$\alpha$-center proximity}~\cite{ABS12} for~$k$ if in the optimum $k$-clustering the distance of each data point to any center except for its own is larger by a factor of at least~$\alpha$ than the distance to its own center.

We apply these notions to hierarchical clustering by showing that if there is a well-separated optimum solution for a level, then the clustering computed by Ward on this level is a $2$-approximation. This means that Ward finds good clusterings for all levels of granularity that have a meaningful clustering; and these good clusterings have a  hierarchical structure. For levels on which the sizes of the optimal clusters are additionally to some extend balanced, we prove that Ward even computes the optimum clustering.

\paragraph*{Related work.}
The design of hierarchical clustering algorithms that satisfy per-level guarantees started with the paper by Dasgupta and Long~\cite{DasguptaL05}. They give a deterministic $8$-approximation and a randomized $2e$-approximation for hierarchical $k$-center. 
Their method turns Gonz\'alez' algorithm~\cite{G85} into a hierarchical clustering algorithm. 
Gonz\'alez' algorithm is a $2$-approximation not only for $k$-center, but also for the \emph{incremental} $k$-center problem: Find an ordering of all points, such that  for all $k$, the first $k$ points in the ordering approximately minimize the $k$-center cost.
The idea to make an algorithm for incremental clustering hierarchical was picked up by Plaxton~\cite{P06}, who proves that this approach leads to a constant factor approximation for the hierarchical $k$-median problem. He uses an incremental $k$-median algorithm due to Mettu and Plaxton~\cite{MP03}.
Finally, Lin, Nagarajan, Rajaraman and Williamson~\cite{LNRW10} propose a general framework for approximating incremental problems that also works for incremental variants of MST, vertex cover, and set cover. They also cast hierarchical $k$-median and $k$-means into their framework for incremental approximation. They get a randomized/deterministic $20.06/41.42$-approximation for hierarchical $k$-median and a randomized/deterministic $151.1\alpha / 576\alpha$-approximation for $k$-means, where $\alpha$ is the approximation ratio of a $k$-means approximation algorithm. Thus, applying~\cite{ANSW17} yields guarantees of $961$ and $3662$, respectively.

Lattanzi, Leonardi, Mirrokni, and Razenshteyn~\cite{LLMR15} develop a constant factor algorithm for \emph{robust} hierarchical $k$-center, i.e., a variant with outliers.
In a different line of work, Dasgupta recently developed a new cost function for similarity-based hierarchical clusterings~\cite{D16}. Although it can be transferred to the setting of dissimilarity measures, this yields an objective for which \emph{any} solution is a constant factor approximation~\cite{CKMM18}. Work on this new cost function includes~\cite{CC17,CKMM18,D16}. Balcan et al.\ present an algorithm for computing hierarchical clusterings that clusters the data accurately in the presence of outliers if the data satisfies certain clusterability properties~\cite{BBV08,BLG14}.

In practice, $k$-means and hierarchical $k$-means are rather tackled by popular heuristics, but the properties of these algorithms are often unknown. 
The famous $k$-means algorithm due to Lloyd~\cite{L82} was analyzed about ten years ago and became the subject of many papers, including~\cite{AMR09,AV06,AV09,D03,MR09,ORSS12,V11}. This has led to the development of $k$-means++~\cite{AV07}, a practically efficient algorithm with a theoretical approximation guarantee of $\mathcal{O}(\log k)$.

Hierarchical clustering algorithms work either top-down (divisive methods) or bottom-up (agglomerative methods). 
Agglomerative methods are more popular because they are usually faster, and the most popular agglomerative methods are based on the complete linkage strategy.
Here, the clusters to be merged are those which minimize the cost of the clustering in the next step. 
Using complete linkage for $k$-means yields Ward's method~\cite{W63}.

There is a relatively small number of papers studying the performance of complete linkage algorithms. Dasgupta and Long~\cite{DasguptaL05} establish the above mentioned $\log k$ lower bound for $k$-center. Ackermann, Bl{\"{o}}mer, Kuntze, and Sohler~\cite{ABKS14} study complete linkage for variants of $k$-center \emph{in the Euclidean space}. The variants include minimizing the radius, the discrete radius and the diameter. 
They show that for constant dimension, complete linkage provides $\mathcal{O}(\log k)$-approximations for $k$-center as well as all variants of it. The drawback is that the approximation factor depends on the the dimension of the space (the extent of the dependence goes from linear dependence to doubly exponential dependence, depending on the problem variant). Großwendt and Röglin~\cite{GR17} improve the analysis, showing that for constant dimension, complete linkage indeed provides an $\mathcal{O}(1)$-approximation. The dependencies on $d$ prevail.

Balcan, Liang, and Gupta~\cite{BLG14} observe that Ward's method cannot be used to recover a given target clustering.

There is a vast body of literature on clusterability assumptions, i.e., assumptions on the input that make clustering easier either in the sense that a target clustering can be (partially) recovered or that a good approximation of an objective function can be computed efficiently. A survey of recent work in this area can be found in~\cite{B15}. Particularly relevant for our paper are the notions of $\delta$-center separation~\cite{Ben-DavidH14} and $\alpha$-center proximity~\cite{ABS12} mentioned above. 
 There are several papers showing that under these assumptions it is possible to recover the target/optimal clustering if~$\delta$ and $\alpha$ are sufficiently large~\cite{ABS12,BL16,KSB16,MM16}. Other notions include the \emph{strict separation property} of Balcan, Blum, and Vempala~\cite{BBV08}, the \emph{$\epsilon$-separation property} of Ostrovsky et al.~\cite{ORSS12}, and the weaker version of the proximity condition due to Kumar and Kannan~\cite{KK10} which Awasthi and Sheffet~\cite{AS12} proposed (it is based on the spectral norm of a matrix whose rows are the difference vectors between the points in the data set and their centers). For all these notions of clusterability, algorithms are developed that (partially) recover the target/optimal clustering.

\paragraph*{Our results.}
In §\ref{sec:WellClusterable}, we analyze the approximation factor of Ward's method on data sets that satisfy different well-known clusterability notions. It turns out that the assumption that the input satisfies a high \emph{$\delta$-center separation}~\cite{Ben-DavidH14} or \emph{$\alpha$-center proximity}~\cite{ABS12} implies a very good bound on the approximation guarantee of Ward's method. We show that Ward's method computes a $2$-approximation for all values of~$k$ for which the input data set satisfies $(2+2\sqrt 2)$-center separation or $(3+2\sqrt 2)$-center proximity. We also show that on instances that satisfy $(2+2\sqrt{2 \nu})$-center separation and for which all clusters~$O_i$ and~$O_j$ in the optimal clustering satisfy~$|O_j|\ge |O_i|/\nu$, Ward even recovers the optimal clustering.

In §\ref{sec:explowerbound} we show that, in general, Ward's method does not achieve a constant approximation factor. We present a family of instances~$(P_d)_{d\in\NN}$ with~$P_d\subset\RR^d$ on which the cost of the $2^d$-clustering computed by Ward is larger than the cost of the optimal $2^d$-means clustering of~$P_d$ by a factor of~$\Omega((3/2)^d)$. 
Then we observe that the family of instances used for this lower bound satisfy the \emph{strict separation property} of Balcan, Blum, and Vempala~\cite{BBV08}, the $\epsilon$-separation property of Ostrovsky et al.~\cite{ORSS12} for any~$\epsilon>0$, and the separation condition from Awasthi and Sheffet~\cite{AS12}. Hence, none of these three notions of clusterability helps Ward's method to avoid that the approximation factor grows exponentially with the dimension.

Finally in §\ref{section:ward:1d} we show that the approximation ratio of Ward's method on one-dimensional inputs is~$\mathcal{O}(1)$. The one-dimensional case turns out to be more tricky than one would expect, and our analysis is quite complex and technically challenging.

\paragraph*{Preliminaries.}
We consider inputs in the Euclidean space $\mathbb{R}^d$. 
The Euclidean distance of $x_1, x_2 \in \mathbb{R}^d$ is denoted by $||x_1-x_2||=||x_1-x_2||_2$.
Let $P \subset \mathbb{R}^d$ be a finite set of points.
For any center $c \in \mathbb{R}^d$, we denote the sum of the squared distances of each point in $P$ to $c$ by 
$\Delta (P,c)= \sum_{p \in P} ||p-c||^2.$ 
This sum is minimized when the center is the \emph{centroid} $\mu(P) := \frac{1}{|P|} \sum_{p \in P} p$ of $P$. We set $\Delta(P):=\Delta(P,\mu(P))$.
For any set of $k$ centers $C \subset \mathbb{R}^d$, the \emph{$k$-means objective cost} is 
$
\Delta (P,C) = {\sum_{p \in P}} \min_{c \in C} ||p-c||^2.
$
The $1$-means cost of $P$ is $\Delta(P)$.
If $P$ is weighted with a weight function $w:P\to \mathbb{N}_{\ge 1}$, then we denote the \emph{total weight} by $w(P) := \sum_{x\in P} w(x)$ and extend the above notations by $\mu(P,w) = \frac{1}{w(P)} \sum_{x\in P} w(x) x$, $\Delta(P,w,c) = \sum_{x\in P} w(x) ||x-c||^2$, and $\Delta(P,w) = \Delta(P,w,\mu(P,w))$. The weighted $k$-means objective is $\Delta (P,w,C)=\sum_{x \in P} \min_{c \in C} w(x)||x-c||^2$. We denote by $\opt_k(P)$ / $\opt_k(P,w)$ the value of a solution that minimizes the (weighted) $k$-means objective, i.e., $\opt_k(P)=\min_{C\subset\RR^d,|C|=k}\Delta(P,C)$ and $\opt_k(P,w)=\min_{C\subset\RR^d,|C|=k}\Delta(P,w,C)$, respectively.

We use the abbreviation $[i] = \{1,\ldots,i\}$ for $i \in \mathbb{N}$.

\subparagraph*{Hierarchical clustering.}
As described by Dasgupta and Long~\cite{DasguptaL05}, a \emph{hierarchical clustering} is a nested partitioning of a point set $P$ into $1,2,3,\ldots$ and finally $n$ clusters, where each intermediate clustering is a more fine-grained version of the previous clustering that results from dividing one cluster into two.
This definition is \lq top-down\rq. Complete linkage algorithms 
build the hierarchical clustering \lq bottom-up\rq\ by starting with $n$ singleton clusters and then subsequently merging two clusters into one until only one cluster remains. We will adapt this view and define a hierarchical clustering $\mathcal{H}$ as a sequence of partitionings $\mathcal{H}_0,\ldots,\mathcal{H}_{n-1}$, where $\mathcal{H}_0 = \{ \{x\} \mid x \in P\}$ and $\mathcal{H}_{n-1} = \{ P\}$, i.e., $\mathcal{H}_i$ shall be the clustering after $i$ merges. The intermediate partitionings satisfy that 
$\mathcal{H}_i = \mathcal{H}_{i-1} \backslash \{A_i,B_i\} \cup \{A_i \cup B_i\}$
for two clusters $A_i, B_i \in \mathcal{H}_{i-1}$. 
Note that we can fully describe $\mathcal{H}$ by the sequence of the $n-1$ \emph{merge operations}  
$
(A_1,B_1),(A_2,B_2),\ldots,(A_{n-1},B_{n-1})
$
that it implicitly contains.

A hierarchical clustering contains a $k$-clustering for any $k \in \{1,\ldots,n\}$. The clusterings are given as partitionings, the centers are implicitly defined as the centroids. More precisely, the $k$-clustering defined by a hierarchical clustering $\mathcal{H}$ has the centers $\{ \mu(Q) \mid Q \in \mathcal{H}_{n-k}\}$. We thus define the $k$-means clustering cost of $\mathcal{H}$ for a given $k$ as
\[
\Delta_k(\mathcal{H}) = \sum_{Q \in \mathcal{H}_{n-k}} \Delta(Q,\mu(Q)) = \sum_{Q \in \mathcal{H}_{n-k}} \Delta(Q).
\]

\subparagraph*{Useful Facts about \texorpdfstring{$k$}{k}-means.}

The following two facts are well known.

\begin{lemma}[Relaxed triangle inequality]\label{lem:ti:two}
For all $x,y,z \in \RR^d$, 
$
||x-y||^2 \le 2(||x-z||^2 + ||z-y||^2).
$
\end{lemma}

\begin{lemma}\label{magicformula}
	For any finite point set $P \subset \mathbb R^d$ and any $c \in \mathbb R^d$, $\Delta (P,c) = \Delta(P)+|P|\cdot||c-\mu(P)||^2.$
\end{lemma}

Lemma~\ref{magicformula} has the following important consequence. Whenever a set of points $P'$ is clustered \emph{together}, i.e., all points in it are assigned to the same center in a given solution, then the cost for this assignment can be computed by knowing only the centroid of the point set and $\Delta(P')$. 
Thus, we can treat such a $P'$ as one weighted point with some additional constant cost. This view is very helpful to simplify the analysis of agglomerative hierarchical clustering strategies. 

\subparagraph*{Ward's method.}
\emph{Ward's method} (or simply \emph{Ward} in the following) is a greedy algorithm. To describe it, the easiest way is to define the following quantity that describes how much the sum of the $1$-means costs increases when merging two clusters. 

\begin{definition}
	Let $A, B \subset \mathbb{R}^d$ be two finite point sets. We define 
	$D(A,B)=\Delta(A \cup B) - \Delta(A) - \Delta(B)$.	
	If a set contains only one point, e.g., $A=\{a\}$, we slightly abuse notation and write $D(a,B) = D(\{a\},B)$ (similarly, if $A=\{a\}$ and $B=\{b\}$, we write $D(a,b)=D(\{a\},\{b\})$).
\end{definition}

Ward's method is agglomerative. It starts with $n$ singleton clusters. Then in every step, it greedily chooses two clusters $A, B$ in the current clustering for which $D(A,B)$ is minimal. This choice is optimal for the next clustering, but subsequent merges and clusterings may suffer from it. We denote the costs of the $k$-clustering computed by Ward's method on data set~$P$ by~$\mathrm{Ward}_k(P)$.

\section{Techniques and Observations}

\subsection{Upper Bounds: Proof Technique in a Nutshell}\label{subsec:ProofTechnique}

Let us give an overview of the basic idea underlying our proof that Ward's method computes a $2$-approximation for all values of~$k$ for which the input data set satisfies $(2+2\sqrt 2)$-center separation or $(3+2\sqrt 2)$-center proximity. The main challenge is to relate the cost of the $k$-clustering computed by Ward to the cost of an optimal $k$-clustering. For this, we fix an arbitrary optimal $k$-clustering~$O_1,\ldots,O_k$. Consider an arbitrary cluster~$O_j$ and let~$P_1^j,\ldots,P_{n_j}^j$ be the data points~$O_j$ consists of (in the actual proof, $P_i^j$ is defined slightly differently). We consider the set $\mathcal{S}_j=\{\{P_1^j,P_2^j\},\{P_2^j,P_3^j\},\ldots,\{P_{n_j-1}^j,P_{n_j}^j\}\}$ of merges. Observe that the merges in~$\mathcal{S}_j$ cannot be applied one after another because after the first merge~$\{P_1^j,P_2^j\}$ the singleton point~$P_2^j$ is gone, which is to be merged in the second merge~$\{P_2^j,P_3^j\}$. Since it is possible to do every second merge of $\mathcal{S}_j$, one can argue that all merges in~$\mathcal{S}_j$ together cost at most~$2\Delta(O_j)$. Now let~$\mathcal{S}=\cup_j \mathcal{S}_j$. Then all merges in~$\mathcal{S}$ together cost at most~$2\opt_k$.

The next step is then to construct a bijection between the set~$\mathcal{S}_{\mathrm{Ward}}$ of the~$n-k+1$ merges performed by Ward to form a $k$-clustering and the set~$\mathcal{S}$. This bijection has the property that every merge of Ward is at most as expensive as the merge in~$\mathcal{S}$ assigned to it. This implies that Ward computes a clustering with cost at most~$2\opt_k$. In order to construct this bijection, consider a step of Ward in which two clusters~$A$ and~$B$ are merged. Let~$\mathcal{C}$ denote the current clustering directly before this merge happens, and let~$\mathcal{S}_{\mathcal{C}}\subseteq\mathcal{S}$ denote the set of those merges from~$\mathcal{S}$ that are feasible in~$\mathcal{C}$ and unassigned, i.e., those merges for which both clusters are contained in~$\mathcal{C}$ and that have not been assigned to any previous merge of Ward. We know that any merge from~$\mathcal{S}_{\mathcal{C}}$ is at least as expensive as the merge of~$A$ and~$B$ because Ward chooses the next merge greedily. Hence, in the bijection we can map the merge of~$A$ and~$B$ to an arbitrary merge from~$\mathcal{S}_{\mathcal{C}}$. This implies that if~$\mathcal{S}_{\mathcal{C}}$ is non-empty in every step, the bijection can be constructed. Since~$|\mathcal{S}|=|\mathcal{S}_{\mathrm{Ward}}|$ this can only be guaranteed if every merge of Ward decreases the number of available merges in~$\mathcal{S}$ by only~1. One can show that this follows from the separation assumption.

For the one-dimensional case, the basic approach is similar. The main difference is that without separation, we can no longer guarantee that the number of available merges decreases by only~$1$ with every step of Ward. Indeed, the original set~$\mathcal{S}$ of good merges may be empty after $n-2k$ merges. To bound the cost of the remaining merge steps, we find a new set of (relatively) good merges, i.e., a set of merges whose costs can be bounded by a constant times~$\opt_k$. Again, this set may run dry, and we have to start again. Essentially, we show that after a constant number of \emph{phases} (Ward merges that are charged against a specific set of good merges), Ward has obtained a $k$-clustering. 

Although the basic idea is similar, the technical implementation of the proof for $d=1$ is very different from our proof for well-clusterable data. Every time that Ward does not merge in a way compatible to the optimum clustering, we have to account for all possible consequences. Techniques like reordering help us to organize the proof. We also simplify the instance before the actual proof.

\subsection{Useful Statements}\label{techniques}

Here we discuss some of the technical statements which we feel may be of interest for future work. All omitted proofs in this section can be found in the full version of this paper.

\paragraph*{Cost of one step.}

The value $D(A,B)$ plays a central role in the analysis of Ward's method. 
By using Lemma~\ref{magicformula}, it is easy to show that $D(A,B)$ does not depend on $\Delta(A)$ or $\Delta(B)$. The following
lemma gives an explicit formula, which leads to convenient upper and lower bounds. These bounds say that the cost of merging two clusters is roughly equivalent to assigning the points of the smaller cluster to the centroid of the larger cluster. 

\begin{restatable}{lemma}{dabcomputation}\label{obsb}
		Let $A$ and $B$ be two clusters. Then 
		$D(A,B)=\frac{|A||B|}{|A|+|B|} \cdot ||\mu_A-\mu_B||^2$.
		Furthermore,
		$\frac{1}{2} \cdot\min\{|A|,|B|\} \cdot ||\mu_A-\mu_B||^2 \le D(A,B) \le \min\{|A|,|B|\} \cdot ||\mu_A-\mu_B||^2.$
		The left hand side is attained for $|A|=|B|$, and the right hand side for $\frac{\max\{|A|,|B|\}}{\min\{|A|,|B|\}} \to \infty $.
\end{restatable}

\paragraph*{How cost accumulates.}

Notice that whenever Ward makes a decision, it is optimal for the clustering in the next step. Where does its error lie? The problem is that every merge forces the points of two clusters to be in the same cluster for any clustering to come. In later clusterings, the condition to cluster certain points together may induce error. We need a way to bound this error. We prove the following technical statement.

\begin{restatable}{corollary}{goodmergethree}\label{cor:goodmergewiththree}
Let $A$, $B$, and $C$ be three disjoint sets of points with $|A| \le |B|$ (or $w(A) \le w(B)$, for weighted sets). Then
$
\Delta(A\cup B\cup C)
\le \Delta(A) + 3 \cdot \Delta(B\cup C) + 4 \cdot D(A,B)
$
and
$
D(A\cup B, C) \le 3 \cdot \Delta(B\cup C) + 3 \cdot D(A,B) - \Delta(B) - \Delta(C)
$.
\end{restatable}

To see how Corollary~\ref{cor:goodmergewiththree} can be used, assume that $A \subset O_i$ and $B\subset O_j$ belong to different optimum clusters which Ward merged during its execution.
Now Corollary~\ref{cor:goodmergewiththree} tells us something about the compatibility of $A \cup B$ with the optimum clustering. 
We pick the smaller of the two clusters, say $A$. Assume that we still have some subset of $B$'s optimum cluster, i.e., there is a cluster $C \subset O_j$ that is still part of the clustering. Then we can merge $A\cup B$ with $C$. Corollary~\ref{cor:goodmergewiththree} says that what we lose is proportional to the optimum cost plus the cost that we already invested into our clustering at an earlier time:
$\Delta(A)$ and $\Delta(B \cup C)$ are both part of the optimum cost, and $D(A,B)$ is what Ward (accumulatively) already payed for merging $A$ and $B$. 

\paragraph*{Monotonicity.}
Notice that performing arbitrary merge operations is not monotone: Say that $a < b < c$ are one-dimensional points such that the centroid of $a$ and $c$ is $b$. Then merging $a$ and $c$ first results in a point set where merging with $b$ costs nothing; clearly, this is not monotone. 
Indeed, when considering a natural variant of Ward's method for the related $k$-median problem, monotonicity is not true. Even for a simple isosceles triangle, greedily chosen merges result in non-monotone merge costs.
However, Ward's merges are indeed monotone. We show the following statement by proving a decomposition lemma for $D(A,B)$.

\begin{restatable}{corollary}{monotonicitycor}[Monotonicity of Ward's method]\label{cor:monotonicity}
	Let $D_i$ be the increase of the objective function in the $i$-th step of Ward's method. Then $D_i \leq D_j$ for $i \leq j$.
\end{restatable}

Monotonicity is a very helpful property. In the argument discussed in §\ref{subsec:ProofTechnique} we use, e.g., that all merges that are possible in the final $k$-clustering computed by Ward's method are at least as expensive as all merges that are performed before by Ward's method to obtain the $k$-clustering.

\paragraph*{Special structures in dimension one.}
The following statements only hold for $d=1$. First we observe that Ward satisfies the following convexity property. 
\begin{restatable}[Convexity in $\mathbb{R}^1$]{lemma}{convexlemma}\label{konvex}
	For any three finite convex clusters $A, B, C \subset \mathbb{R}^1$ with $\mu(A) < \mu(C) < \mu(B)$, we have $D(A,C) < D(A,B)$ or $D(B,C)<D(A,B)$.	
\end{restatable}

Lemma~\ref{konvex} means that Ward will never merge $A$ and $B$ if a point or cluster lies between them on the line. This establishes that Ward's clusters never overlap. It gives us a concept of neighbors on the line. 

We combine Lemma~\ref{konvex} with a convexity property of Ward (see Corollary~\ref{subcluster}). This allows us to prove a powerful technique that we call \emph{reordering}. 
Say that Ward at some point merges two clusters $A$ and $B$. Then $A$ and $B$ are neighbors on the line. This means that merging $A$ and $B$ will result in a centroid $\mu(A\cup B)$ which is further away from any other cluster than $\mu(A)$ and $\mu(B)$ are. So, clusters that did not want to merge with $A$ or $B$ would also not merge with $A\cup B$ (by Corollary~\ref{subcluster}). Thus, we could perform the merge $(A,B)$ \emph{earlier} without distorting Ward's course of action at all (except that the merge $(A,B)$ is at the wrong position). This allows us to reorder Ward's merges for our analysis.

%!TEX root = wards-soda-final.tex

\section{Ward on Well-Clusterable Data}\label{sec:WellClusterable}

Clustering suffers from a general gap between theoretical study and practical application; clustering objectives are usually NP-hard to optimize, and even NP-hard to approximate to arbitrary precision. On the other hand, heuristics like Lloyd's algorithm, which can produce arbitrarily bad solutions, are known to work well or reasonably well in practice. One way of interpreting this situation is that data often has properties that make the problem computationally easier. Indeed, for clustering it is very natural to assume that the data has some structure -- otherwise, what do we hope to achieve with our clustering? The challenge is to find good measures of structure that characterize what makes clustering easy (but non-trivial). 

Many notions of \emph{clusterability} have been introduced in the literature and there are also different ways to measure the quality of a clustering. While traditionally a clustering is evaluated on the basis of an objective function (e.g., the $k$-means objective function), there has been an increased interest recently to study which notions of clusterability make it feasible to recover (partially) a \emph{target clustering}, some \emph{true} clustering of the data. For this, the niceness conditions imposed on the input data are usually some form of separation condition on the clusters of the target clustering. We study the effect of five well-studied clusterability notions on the quality of the solution computed by Ward's method. 

First we study the notions of \emph{$\delta$-center separation} and \emph{$\alpha$-center proximity}, which have been introduced by Ben-David and Haghtalab~\cite{Ben-DavidH14} and Awasthi, Blum, and Sheffet~\cite{ABS12}, respectively. 

\begin{definition}[\cite{Ben-DavidH14}]
An input $P \subset \mathbb{R}^d$ satisfies \emph{$\delta$-center separation} with respect to some target clustering $C_1,\ldots,C_k$ if there exist centers $c_1^\ast,\ldots,c_k^\ast \in \mathbb{R}^d$ such that
$||c_j^\ast - c_i^\ast|| \ge \delta \cdot \max_{\ell\in [k]} \max_{x\in C_\ell} ||x-c_{\ell}^\ast||$
for all $i \neq j$.
We say the input satisfies \emph{weak $\delta$-center separation} if for each cluster $C_j$ with $j \in [k]$ and for all $i \neq j$, $||c_j^\ast - c_i^\ast|| \geq \delta \cdot  \max_{x\in C_j} ||x-c_j^\ast||$.
\end{definition}

Kushagra, Samadi, and Ben-David~\cite{KSB16} show that single linkage and a pruning technique are sufficient to find the target clustering under the condition that the data satisfies $\delta$-center separation for $\delta \ge 3$.

While the goal of Ben-David and Haghtalab~\cite{Ben-DavidH14} is to recover a target clustering, we focus in this paper on approximating the $k$-means objective function. Hence, in the following we will always assume that the target clustering $C_1,\ldots,C_k$ is an optimal $k$-means clustering (which we usually denote by~$O_1,\ldots,O_k$) and the centers $c_1^\ast,\ldots,c_k^\ast \in \mathbb{R}^d$ are the optimal $k$-means centers for this clustering. We will make this assumption also for all other notions of clusterability that are based on a target clustering and that we introduce in the following.

\begin{definition}[\cite{ABS12}]
An instance $P$ satisfies \emph{$\alpha$-center proximity} if there exists an optimal $k$-means clustering $O_1,\ldots,O_k$ with centers $c_1^\ast,\ldots,c_k^\ast \in \mathbb{R}^d$ such that for all $j \neq i, j \in [k]$ and for any point $x \in C_i$ it holds $||x-c_j^\ast|| \ge \alpha ||x-c_i^\ast||$.
\end{definition}

Awasthi, Blum, Sheffet~\cite{ABS12} introduced the notion of \emph{$\alpha$-perturbation resilience} and showed that it implies $\alpha$-center proximity. They show that for $\alpha\ge 3$, the optimal clustering can be recovered if the data is $\alpha$-perturbation resilient. This was improved by Balcan and Liang~\cite{BL16} and finally by Makarychev and Makarychev~\cite{MM16}, who show that it is possible to completely recover the optimal clustering for $\alpha=2$. The latter paper shows that the results even hold for a weaker property called \emph{metric perturbation resilience}. We show that for large enough~$\delta$ and~$\alpha$, Ward's method computes a $2$-approximation if the data satisfies $\delta$-center separation or $\alpha$-center proximity.

\begin{theorem}
Let $P \subset \mathbb{R}^d$ be an instance that satisfies weak $(2+2\sqrt 2+\epsilon)$-center separation or $(3+2\sqrt 2+ \epsilon)$-center proximity for some $\epsilon>0$ and some number~$k$ of clusters. Then the $k$-clustering computed by Ward on~$P$ is a $2$-approximation with respect to the $k$-means objective function.
\end{theorem}

We also show that on instances that satisfy $(2+2\sqrt{2 \nu}+\epsilon)$-center separation and for which all clusters~$O_i$ and~$O_j$ in the optimal clustering satisfy~$|O_j|\ge |O_i|/\nu$, Ward even recovers the optimal clustering.

It is interesting to note that the example proposed by Arthur and Vassilvitskii~\cite{AV07} that shows that the famous $k$-means++ algorithm has an approximation ratio of $\Omega(\log k)$ satisfies $\delta$-center separation and $\alpha$-center proximity for large values of~$\delta$ and~$\alpha$, and has balanced clusters, i.e., $\nu=1$.

\begin{observation}
There is a family of examples where $k$-means++ has an expected approximation ratio of $\Omega(\log k)$, while Ward computes an optimal solution. 
\end{observation}

In contrast we will see that the instances that we use to prove our exponential lower bound on the approximation factor of Ward's method (Theorem~\ref{thm:ExponentialLowerBound}) satisfy $\delta$-center separation and $\alpha$-center proximity for~$\delta\le 1+\sqrt{2}$ and~$\alpha\le 1+\sqrt{2}$. We will also see that even for arbitrary large~$\delta$ and~$\alpha$ there are instances that satisfy $\delta$-center separation and $\alpha$-center proximity and on which Ward's method does not compute an optimal solution.
In addition to center separation and center proximity we study the following three other prominent notions of clusterability: the strict separation property due to Balcan, Blum, and Vempala~\cite{BBV08}, $\epsilon$-separation due to Ostrovsky et al.~\cite{ORSS12}, and the separation condition from Awasthi and Sheffet~\cite{AS12}
 We will see that the exponential lower bound instances satisfy these clusterability notions when the target clustering is the optimal $k$-means clustering. Hence, none of these notion guarantees that Ward's method computes a good clustering. 

\begin{corollary}
For any $\epsilon > 0$, there is a family of point sets $(P_d)_{d\in \mathbb{N}}$ with~$P_d\subset\RR^d$ that are $\epsilon$-separated and that satisfy $1+\sqrt{2}$-center separation, $1+\sqrt{2}$-center proximity, the strict separation property and the AS-center separation property where $\mathrm{Ward}_{k}(P_d) \in \Omega((3/2)^d\cdot \mathrm{\opt}_{k}(P_d))$ for~$k=2^d$.
Furthermore, for any $\delta > 1$ and any $\alpha > 1$, there exists a point set that satisfies $\delta$-center separation and $\alpha$-center proximity and for which Ward does not compute an optimal solution.
\end{corollary}

\subsection{Upper Bounds\label{sec:positive-results}}

In this section, we analyze the behavior of Ward on $\delta$-center separated instances and instances that satisfy $\alpha$-center proximity for some number~$k$ of clusters. We are only interested in the $k$-clustering computed by Ward. Hence, in the following we assume that $k$ is fixed and that Ward stops as soon as it has obtained a $k$-clustering. First we prove that center proximity implies weak center separation. Hence, it suffices to study instances that satisfy weak center separation. Deferred proofs can be found in the full version of this paper.

\begin{restatable}{lemma}{RelationProximitySeparation}\label{lemma:RelationProximitySeparation}
Let~$P\subset\RR^d$ be an instance that satisfies $\alpha$-center proximity. Then $P$ also satisfies weak $(\alpha-1)$-center separation.
\end{restatable}

In the following we call a cluster $A$ that is formed by Ward an \emph{inner cluster} if $A$ is completely contained within an optimum cluster. We start our analysis with the following lemma, which states one very crucial property of Ward's behavior on well-separated data. It implies that Ward does not merge inner clusters from two different optimal clusters as long as there exists more than one inner cluster in both of these optimal clusters. 

\begin{restatable}{lemma}{onedeletion}\label{lem:onedeletion}
Let $P \subset \mathbb{R}^d$ be an instance that satisfies weak $(2+2\sqrt 2+ \epsilon)$-center separation for some $\epsilon>0$. Assume we have two optimal clusters $O_1$ and $O_2$ and each of them contains at least two inner clusters $A_1,B_1$ and $A_2,B_2$, respectively, directly after the $i$-th step of Ward. Then, in step~$i+1$, Ward will not merge an inner cluster of $O_1$ with an inner cluster of $O_2$. 
\end{restatable}

\paragraph*{Inner-cluster merges}

In the following, assume that~$P\subset\RR^d$ is an arbitrary instance and and that the clusters $O_1,\dots,O_k$ are an optimal $k$-clustering of~$P$ with objective value~$\opt=\opt_k(P)$. Our goal is to show that the $k$-clustering~$W_1,\ldots,W_k$ computed by Ward on~$P$ is worse by only a factor of at most~$2$ if $P$ satisfies weak $(2+2\sqrt 2+ \epsilon)$-center separation for some $\epsilon>0$. 

Observe that Lemma~\ref{lem:onedeletion} does not exclude the possibility that Ward performs inner-cluster merges on~$P$, i.e., it might merge two inner clusters from the same optimum cluster at some point during its execution. While we will see that in the one-dimensional case one can assume that such inner-cluster merges do not happen, we cannot make this assumption in general.
In our analysis, we bound the costs of the inner-cluster merges separately from the costs of the other merges, which we call \emph{non-inner merges} in the following.

We define an equivalence relation $r$ on $P$ as follows: two points $x_1$ and $x_2 \in P$ are equivalent if and only if there exists an inner cluster $C$ constructed by Ward at some point of time with $x_1,x_2 \in C$. We denote the equivalence classes of~$r$ by $P/r= \lbrace C_1, \dots, C_m \rbrace$. The following observation is immediate.

\begin{observation}\label{obs:fullclass}
If Ward merges in any step an inner cluster $C$ with another cluster that is not an inner cluster of the same optimal cluster, then $C \in P/r$ is an equivalence class.
\end{observation}

This means that the equivalence classes represent inner clusters of Ward right before they are merged with points from outside their optimal cluster. With other words, if we perform all inner cluster merges that are performed by Ward and leave out all non-inner merges, we get the clustering represented by~$P/r$.

Consider an arbitrary optimal cluster~$O_j$ and let~$P_1^j,\ldots,P_{n_j}^j$ denote the inner clusters of~$O_j$ in~$P/r$. We assume that these inner clusters are indexed in the order in which they are merged with other clusters by Ward. To illustrate this definition, consider the step in which~$P_i^j$ is merged by Ward with some other cluster~$Q$. Since~$P_i^j\in P/r$, this step is a non-inner merge and in particular~$Q$ is not equal to any of the clusters~$P_{i+1}^j,\ldots,P_{n_j}^j$. At the time this merge happens, the indexing guarantees that the cluster~$P_{i+1}^j$ is either present or there exist multiple parts~$C_1,\ldots,C_{\ell}$ of~$P_{i+1}^j$ that are only later merged by inner-cluster merges to~$P_{i+1}^j$. Since Ward merges~$P_i^j$ and~$Q$, we know that~$D(P_i^j,Q) \leq D(P_i^j,C_h)$ for any~$h\in[\ell]$. We will use this fact to give an upper bound for the costs of the clustering~$W_1,\ldots,W_k$.

It might be that some inner clusters of~$O_j$ in~$P/r$ are not merged at all by Ward and contained in the clustering~$W_1,\ldots,W_k$. These inner clusters are the last in the ordering, i.e., they are~$P_a^j,\ldots,P_{n_j}^j$ where~$n_j-a+1$ is the number of such clusters.

\paragraph*{Potential graph}

In order to bound the costs of the clustering~$W_1,\ldots,W_k$ produced by Ward we introduce the \emph{potential graph}~$G=(V,E)$ with vertex set~$V=P/r$. The edges~$E$ of~$G$ are directed and there are only edges between inner clusters of the same optimal cluster. Consider an arbitrary optimal cluster $O_j$ with $j \in [k]$ and let $P_1^j \ldots P_{n_j}^j$ be the inner clusters of~$O_j$ in~$P/r$ indexed as above in the order in which they are merged with other clusters by Ward. Then for every~$i\in[n_j-1]$ the set~$E$ contains the edge~$(P_i^j,P_{i+1}^j)$. Both the vertices and the edges are weighted and we denote the sum of all vertex and edge weights by $w(G)$.

The weight of a vertex~$Q\in P/r$ is defined as $w(Q)=\Delta(Q)$, i.e., the weight of vertex~$Q$ equals the costs of forming the inner cluster~$Q$. We will now define weights for the edges such that the sum of all vertex and edge weights in the potential graph is at most~$2\opt_k$. After that we prove that there is a one-to-one correspondence between the non-inner merges of Ward and the edges in the graph such that the costs of each non-inner merge of Ward are at most the weight of the associated edge. Together this proves that Ward computes a solution with costs at most~$2\opt_k$.

To define the weight of the edge~$(P_i^j,P_{i+1}^j)$, we first consider the case that~$P_i^j$ is merged at some point of time with another cluster~$Q$ by Ward. Then let~$C_1,\ldots,C_{\ell}$ again denote the parts of~$P_{i+1}^j$ that are present at that point of time. The edge weight~$w(P_i^j,P_{i+1}^j)$ is defined as~$\max_{h\in[\ell]}D(P_i^j,C_h)$\footnote{When reading the proof the reader might notice that our definition of $w(P_i^j,P_{i+1}^j)$ is to some extend arbitrary. Instead of defining it as~$\max_{h\in[\ell]}D(P_i^j,C_h)$, we could also define it as~$\min_{h\in[\ell]}D(P_i^j,C_h)$ or as $D(P_i^j,C_h)$ for any~$h$.}. Observe that since Ward performs greedy merges, this definition guarantees that the merge of~$P_i^j$ and~$Q$ costs at most the edge weight~$w(P_i^j,P_{i+1}^j)$. If~$P_i^j$ is not merged at all by Ward, we set the weight~$w(P_i^j,P_{i+1}^j)$ to~$D(P_i^j,P_{i+1}^j)$.

\begin{restatable}{lemma}{DeltaPartition}\label{lemma:DeltaPartition}
Let~$P\subset\RR^d$ be a finite point set and let~$Q_1,\ldots,Q_{\ell}$ denote an arbitrary partition of~$P$ into pairwise disjoint parts. Then
$\Delta(P) \ge \Delta(Q_1)+\ldots+\Delta(Q_{\ell})$.
\end{restatable}

\fixtoomuchspaceerror
\begin{restatable}{lemma}{TotalWeightPotentialGraph}\label{lem:TotalWeightPotentialGraph}
The weights in the potential graph satisfy~$w(G)\le 2\opt_k$.
\end{restatable} 

\paragraph*{Bijection between non-inner merges and edges}

We have seen that the sum of the weights in the potential graph is at most~$2\opt_k$. Our goal is now to find a bijection between the non-inner merges of Ward and the edges of the potential graph such that the costs of any non-inner merge are bounded from above by the weight of the edge assigned to it in the bijection. The existence of such a bijection implies that also the costs of the solution~$W_1,\ldots,W_k$ computed by Ward are at most~$2\opt_k$.

Now we construct this bijection. Let us first consider non-inner merges in which at least one of the clusters is an inner cluster contained in~$P/r$. Let this be the inner cluster~$P_i^j$ of some optimal cluster~$O_j$ and assume further that~$i<n_j$. Then~$P_i^j$ has an outgoing edge to~$P_{i+1}^j$. We denote by~$Q$ the cluster with which~$P_i^j$ is merged and we assign the merge of~$P_i^j$ with~$Q$ to the edge~$(P_i^j,P_{i+1}^j)$ in the bijection.

\begin{restatable}{lemma}{onedeletiontwo}
\label{lem:onedeletiontwo}
Let $P \subset \mathbb{R}^d$ be an instance that satisfies weak $(2+2\sqrt 2+ \epsilon)$-center separation for some $\epsilon>0$. Consider a non-inner merge of Ward between two inner clusters from~$P/r$. Then at most one of these inner clusters has an outgoing edge in~$G$.
\end{restatable}

Observe that it cannot happen that the same edge is assigned to two different merges by the construction described above because an edge~$(P_i^j,P_{i+1}^j)$ can only be assigned to a step in which~$P_i^j$ is merged with some other cluster and there can only be one such merge.

Let~$L\subseteq E$ denote the set of edges that are not assigned to a step of Ward by the above construction. The potential graph~$G$ contains~$|V|=|P/r|$ vertices and~$|V|-k$ edges. Since the number of non-inner merges of Ward is also~$|V|-k$, there are also~$|L|$ non-inner merges that are not yet assigned to an edge. We finish the construction of the bijection by assigning the unassigned non-inner merges arbitrarily bijectively to~$L$.

\begin{restatable}{lemma}{NonInnerBijection}\label{lemma:NonInnerBijection}
The costs of each non-inner merge of Ward are bounded from above by the weight of the assigned edge in the potential graph.
\end{restatable}

Now the following theorem follows easily.

\begin{restatable}{theorem}{TwoApproximation}
\label{thm:2Approximation}
Let $P \subset \mathbb{R}^d$ be an instance that satisfies weak $(2+2\sqrt 2+ \epsilon)$-center separation or $(3+2\sqrt 2+ \epsilon)$-center proximity for some $\epsilon>0$. Then Ward computes a $2$-approximation on $P$.
\end{restatable}

\fixtoomuchspaceerror
\begin{restatable}{theorem}{wardoptimalitycriterion}
\label{thm:ward-optimality-criterion}
Let $P \subset \mathbb{R}^d$ be an instance with optimal $k$-means clustering $O_1,\ldots,O_k$ with centers $c_1^\ast,\ldots,c_k^\ast \in \mathbb{R}^d$. Assume that~$P$ satisfies $(2+2\sqrt{2 \nu}+\epsilon)$-center separation for some $\epsilon > 0$, where
$\nu = \max_{i,j \in [k]} \frac{|O_i|}{|O_j|}$
is the largest factor between the sizes of any two optimum clusters.
Then Ward computes the optimal $k$-means clustering $O_1,\ldots,O_k$.
\end{restatable}

In the full version of this paper we show that Theorem~\ref{thm:2Approximation} does not hold for significantly smaller~$\delta$ and~$\alpha$.

%!TEX root = wards-soda-final.tex

\section{Exponential Lower Bound in High Dimension}\label{sec:explowerbound}

In the following, we describe a family of instances of increasing dimension~$d$ where Ward computes for some number~$k=k(d)$ of clusters a $k$-clustering that costs $\Omega((3/2)^d \opt_k)$. Here and in all other worst-case examples, we assume that given a choice between equally expensive merges, Ward chooses the action that leads to a worse outcome. This is without loss of generality because we can always slightly move the points to ensure the outcome we want. However, it greatly simplifies the exposition.

To further simplify the exposition, the below definitions use points of infinite weight and assume that the optimal cluster centers coincide with these infinite weight points. For any finite realization of the example, that is not the case. To ensure that Ward actually behaves like described in the following, we have to move the high weight points by an infinitesimal distance. We do this in the full version of this paper, but for sake of clarity, omit it in the exposition here. Notice that merging a cluster $H$ of infinite weight with a cluster $A$ of finite weight costs $|A|\cdot||\mu(A)-\mu(H)||^2$ by Lemma~\ref{obsb}.

Let~$d$ be given. We construct an instance~$P_d\subseteq\RR^d$ with $2^{d+1}$ points. 
For~$i\ge 2$ let $z_i^2 = \frac{3^{i-2}}{2^{i-1}}$ and define
\[
\begin{split}
   P_d = \{(x_1,\ldots,x_d)\mid x_1\in\{-1,-(\sqrt{2}-1),\sqrt{2}-1,1\},\\ x_i\in\{-z_i,z_i\} \ \forall i\in \{2,\ldots,d\}\}.
\end{split}
\]
All points from~$P_d$ whose first coordinate is~$-1$ or~$1$ have weight~$\infty$ (we call these \emph{heavy points}). All other points have weight~$1$ (we call these \emph{light points}). 
For an illustration of $P_2$ and $P_3$, see Figure~\ref{fig:lowerbound}.
\begin{figure*}
\centering
\begin{tikzpicture}[scale=2.3]
\node [draw,fill,circle,inner sep=0cm,minimum height=0.15cm] at (-1,-0.5) {};
\node [draw,fill,circle,inner sep=0cm,minimum height=0.05cm] at (-0.41421356,-0.5) {};
\node [draw,fill,circle,inner sep=0cm,minimum height=0.05cm] at (+0.41421356,-0.5) {};
\node [draw,fill,circle,inner sep=0cm,minimum height=0.15cm] at (+1,-0.5) {};

\node [draw,fill,circle,inner sep=0cm,minimum height=0.15cm] at (-1,0.5) {};
\node [draw,fill,circle,inner sep=0cm,minimum height=0.05cm] at (-0.41421356,0.5) {};
\node [draw,fill,circle,inner sep=0cm,minimum height=0.05cm] at (+0.41421356,0.5) {};
\node [draw,fill,circle,inner sep=0cm,minimum height=0.15cm] at (+1,0.5) {};

\draw [->] (-1.5,0) -- (1.5,0);
\draw [->] (0,-1) -- (0,1);

\node [draw,rectangle,inner sep=0cm,minimum height=0.2cm,minimum width=0cm, label={[label distance=-0.1cm]above:{\tiny$-1$}},label distance=-0.1cm] at (-1,0) {};
\node [draw,rectangle,inner sep=0cm,minimum height=0.2cm,minimum width=0cm, label={[label distance=-0.1cm]above:{\tiny$+1$}},label distance=-0.1cm] at (1,0) {};
\node [draw,rectangle,inner sep=0cm,minimum height=0.2cm,minimum width=0cm, label={[label distance=-0.1cm]above:{\tiny$-(\sqrt{2}-1)$}},label distance=-0.1cm] at (-0.41421356,0) {};
\node [draw,rectangle,inner sep=0cm,minimum height=0.2cm,minimum width=0cm, label={[label distance=-0.1cm]above:{\tiny$+(\sqrt{2}-1)$}},label distance=-0.1cm] at (+0.41421356,0) {};

\node [draw,rectangle,inner sep=0cm,minimum width=0.2cm,minimum height=0cm, label={[label distance=-0.1cm]right:{\tiny$+z_2$}},label distance=-0.1cm] at (0,0.5) {};
\node [draw,rectangle,inner sep=0cm,minimum width=0.2cm,minimum height=0cm, label={[label distance=-0.1cm]right:{\tiny$-z_2$}},label distance=-0.1cm] at (0,-0.5) {};

\begin{scope}[xshift=3.5cm]
\node (a) [draw,fill,circle,inner sep=0cm,minimum height=0.15cm] at (-1,-0.5,0.75) {};
\node (b) [draw,fill,circle,inner sep=0cm,minimum height=0.05cm] at (-0.41421356,-0.5,0.75) {};
\node (c) [draw,fill,circle,inner sep=0cm,minimum height=0.05cm] at (+0.41421356,-0.5,0.75) {};
\node (d) [draw,fill,circle,inner sep=0cm,minimum height=0.15cm] at (+1,-0.5,0.75) {};
\node (e) [draw,fill,circle,inner sep=0cm,minimum height=0.15cm] at (-1,0.5,0.75) {};
\node (f) [draw,fill,circle,inner sep=0cm,minimum height=0.05cm] at (-0.41421356,0.5,0.75) {};
\node (g) [draw,fill,circle,inner sep=0cm,minimum height=0.05cm] at (+0.41421356,0.5,0.75) {};
\node (h) [draw,fill,circle,inner sep=0cm,minimum height=0.15cm] at (+1,0.5,0.75) {};

\node (aa) [draw,fill,circle,inner sep=0cm,minimum height=0.15cm,gray] at (-1,-0.5,-0.75) {};
\node (bb) [draw,fill,circle,inner sep=0cm,minimum height=0.05cm,gray] at (-0.41421356,-0.5,-0.75) {};
\node (cc)[draw,fill,circle,inner sep=0cm,minimum height=0.05cm,gray] at (+0.41421356,-0.5,-0.75) {};
\node (dd) [draw,fill,circle,inner sep=0cm,minimum height=0.15cm,gray] at (+1,-0.5,-0.75) {};
\node (ee) [draw,fill,circle,inner sep=0cm,minimum height=0.15cm,gray] at (-1,0.5,-0.75) {};
\node (ff) [draw,fill,circle,inner sep=0cm,minimum height=0.05cm,gray] at (-0.41421356,0.5,-0.75) {};
\node (gg) [draw,fill,circle,inner sep=0cm,minimum height=0.05cm,gray] at (+0.41421356,0.5,-0.75) {};
\node (hh) [draw,fill,circle,inner sep=0cm,minimum height=0.15cm,gray] at (+1,0.5,-0.75) {};

\draw [dashed,thin] (a) -- (aa);
\draw [dashed,thin] (b) -- (bb);
\draw [dashed,thin] (c) -- (cc);
\draw [dashed,thin] (d) -- (dd);
\draw [dashed,thin] (e) -- (ee);
\draw [dashed,thin] (f) -- (ff);
\draw [dashed,thin] (g) -- (gg);
\draw [dashed,thin] (h) -- (hh);

\draw [<->] (0,0.5,-0.75) to node [fill=white,circle] {$2z_3$} (0,0.5,0.75);

\end{scope}

\end{tikzpicture}
\caption{Point set $P_d$ from the family of worst-case examples, drawn for $d=2$ and $d=3$. The heavy points are drawn larger.\label{fig:lowerbound}}
\end{figure*}
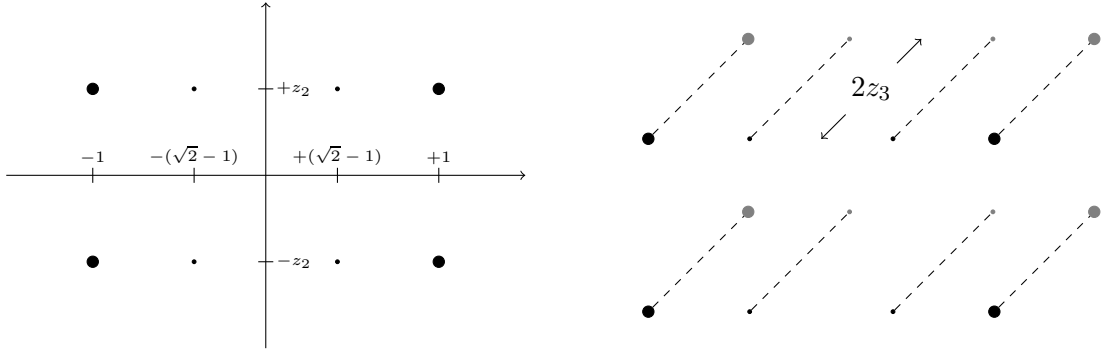

We show the following theorem.
\begin{theorem}\label{thm:ExponentialLowerBound}
The family of point sets $(P_d)_{d\in \mathbb{N}}$ satisfies $\mathrm{Ward}_{k}(P_d) \in \Omega((3/2)^d\cdot \mathrm{\opt}_{k}(P_d))$ for~$k=2^d$.
\end{theorem}

In the theorem, we use~$k=k(d)=2^d$, i.e., we are interested in finding a $2^{d}$-clustering of~$P_d$. Observe that in the optimal $2^{d}$-clustering of~$P_d$, the heavy points are in separate clusters. Due to their infinite weight, they also determine the cluster centers. Hence, in the optimal solution each light point is in the same cluster as its closest heavy point. Since each light point is within distance~$2-\sqrt{2}$ of a heavy point, the cost of the optimal solution is
\[
   \mathrm{\opt}_k(P_d) = 2^{d}\cdot (2-\sqrt{2})^2.
\]

Now we look at a run of Ward's method on~$P_d$. We say that phase~$1$ lasts as long as there is at least one light point that forms its own cluster. We prove by induction that during phase~$1$ the only clusters that occur are singleton clusters consisting of one light or one heavy point and clusters that consist of two light points that differ only in the first coordinate. We call the latter \emph{pair clusters}. At the beginning this is clearly the case. Now assume that the induction hypothesis holds at some point of time in phase~$1$. Merging two heavy points has infinite cost and merging a heavy point with a light point or a pair cluster has cost at least~$(2-\sqrt{2})^2\approx 0.343$ because~$2-\sqrt{2}$ is the minimum distance between a light and a heavy point. Merging two singleton light points that differ only in the first coordinate costs~$\frac12\cdot(2\sqrt{2}-2)^2=(2-\sqrt{2})^2$ (observe that the induction hypothesis guarantees that for any singleton light point the light point that differs only in the first coordinate is also a singleton point). Merging two singleton light points that differ in any other coordinate costs at least~$\frac1{1+1}\cdot (2z_2)^2=1$, merging a singleton light point with a pair cluster costs at least~$\frac{1\cdot 2}{1+2}\cdot (2z_2)^2=\frac43$, and merging two pair clusters costs at least~$\frac{2\cdot 2}{2+2}\cdot (2z_2)^2=2$. Hence, we can assume that Ward merges two singleton light points that differ only in the first coordinate. After that the induction hypothesis is still true. Hence, in phase~$1$ all $2^{d-1}$ pairs of points of the form~$(-(\sqrt{2}-1),x_2,\ldots,x_d)$ and~$(\sqrt{2}-1,x_2,\ldots,x_d)$ will be merged. We call the clusters that consist of these points the $(*,x_2,\ldots,x_d)$-clusters in the following.

Then phase~$2$ starts. Phase~$2$ lasts as long as there are pair clusters. We show by induction that the only clusters that occur in phase~$2$ are singleton heavy points, pair clusters, and clusters with four points that result from merging two pair clusters that differ only in the second coordinate. We call the latter \emph{quadruple clusters}. Merging two pair clusters of the form $(*,-z_2,x_3,\ldots,x_d)$ and $(*,z_2,x_3,\ldots,x_d)$ to form a quadruple cluster costs~$\frac{2\cdot 2}{2+2}(2z_2)^2=2$. Merging two pair clusters that differ in any other coordinate than the second is more expensive because their centers are further apart than~$2z_2$. Merging the $(*,x_2,\ldots,x_d)$-cluster with a heavy point costs at least~$2$ because the center of this cluster is~$(0,x_2,\ldots,x_d)$, which is at distance~$1$ from the heavy points. Similarly merging a quadruple cluster (whose center is~$(0,0,x_3,\ldots,x_d)$) with a heavy point costs at least~$2+z_2^2\ge 2$. Merging a quadruple cluster with a pair cluster costs at least~$\frac{2\cdot 4}{2+4}(2z_3)^3>2$ and merging two quadruple clusters costs at least~$\frac{4\cdot 4}{4+4}(2z_3)^3>2$. Hence, we can assume that Ward merges two pair clusters that differ only in the second coordinate. After that the induction hypothesis is still true. Hence, in phase~$2$ all $2^{d-2}$ pairs of clusters of the form~$(*,-z_2,x_3,\ldots,x_d)$ and~$(*,z_2,x_3,\ldots,x_d)$ will be merged. We call the clusters that consist of these points the $(*,*,x_3,\ldots,x_d)$-clusters in the following.

At the beginning of phase~$i\ge 2$, there are~$2^d$ singleton heavy points and~$2^{d-i+1}$ clusters of the form $(*,\ldots,*,x_i,\ldots,x_d)$ with~$2^{i-1}$ points each. Phase~$i$ ends when there is no cluster of the form $(*,\ldots,*,x_i,\ldots,x_d)$ left. One can show again by induction that Ward merges in phase~$i$ all pairs of clusters of the form $(*,\ldots,*,-z_i,x_{i+1},\ldots,x_d)$ and $(*,\ldots,*,z_i,x_{i+1},\ldots,x_d)$. The center of the cluster~$(*,\ldots,*,x_i,\ldots,x_d)$ is~$(0,\ldots,0,x_{i},\ldots,x_d)$, which is at distance~$\sqrt{1+z_2^2+\ldots+z_{i-1}^2}$ from the heavy points. Hence, merging such a cluster with a heavy point costs at least~$2^{i-1}\cdot (1+z_2^2+\ldots+z_{i-1}^2) = 2^i z_i^2$, where the equation follows from the following observation. 

\begin{restatable}{observation}{zireihe}
\label{obs:zireihe}
It holds that $\nonumber1+z_2^2+\ldots+z_{i-1}^2 = 2 z_i^2$.
\end{restatable}

Merging the clusters~$(*,\ldots,-z_i,x_{i+1},\ldots,x_d)$ and~$(*,\ldots,z_i,x_{i+1},\ldots,x_d)$ costs
\[
   \frac{2^{i-1}\cdot 2^{i-1}}{2^{i-1}+2^{i-1}}\cdot (2z_i)^2 = 2^i z_i^2.
\]
Merging two clusters that differ in one of the~$d-i$ last coordinates costs at least~$\frac{2^{i-1}\cdot 2^{i-1}}{2^{i-1}+2^{i-1}}(2z_{i+1})^2=2^i\cdot z_{i+1}^2>2^iz_i^2$. Hence, in phase~i all $2^{d-i}$ pairs of clusters of the form~$(*,\ldots,*,-z_i,x_{i+1},\ldots,x_d)$ and~$(*,\ldots,*,z_i,x_{i+1},\ldots,x_d)$ will merge, which costs in total $2^{d-i}\cdot 2^i z_i^2$.

Phases $2$ until $d$ together cost $\sum_{i=2}^{d} 2^{d-i}\cdot 2^i z_i^2 = 2^{d}\cdot(2 z_{d+1}^2-1) = 2 \cdot 3^{d-1} - 2^d$, where we used Observation~\ref{obs:zireihe}. After phase~$d$, all light points will be in the same cluster. Then the number of clusters is~$2^d+1$ and in the last step the cluster of light points, whose center is the origin, will be merged with one heavy point. This costs
\[
 2^d\cdot (1+z_2^2+\ldots+z_{d}^2) = 2^{d+1} \cdot z_{d+1}^2 = 2 \cdot 3^{d-1}.
\]
Phase $1$ costs in total $2^{d-1}(2-\sqrt{2})^2$.
Thus, the overall cost of Ward's solution is
\begin{align*}
\mathrm{Ward}_k(P_d) & = 2^{d-1}(2-\sqrt{2})^2 + 2 \cdot 3^{d-1} + 2 \cdot 3^{d-1} - 2^d\\
 & = 4 \cdot 3^{d-1} + 2^{d-1}(2-\sqrt{2})^2 - 2^d.
\end{align*}
This implies
\begin{align*}
\frac{\mathrm{Ward}_k(P_d)}{\mathrm{\opt}_k(P_d)} 
& = \frac{4 \cdot 3^{d-1} + 2^{d-1}(2-\sqrt{2})^2 - 2^d}{2^{d}\cdot (2-\sqrt{2})^2} \\
& = \frac{4}{3 (2-\sqrt{2})^2} \cdot \left(\frac{3}{2}\right)^d +\frac{1}{2} - \frac{1}{(2-\sqrt{2})^2}\\
& \in \Omega\left(\left(\frac{3}{2}\right)^d\right).
\end{align*}

%!TEX root = wards-soda-final.tex

\section{Ward's Method in Dimension One}\label{section:ward:1d}

In this section, we discuss the approximation ratio of Ward's method for inputs $P \subset \mathbb{R}^1$ and show the following theorem.

\begin{theorem}\label{thm:OneDimensional}
Let $P \subset \mathbb{R}$ be an arbitrary instance that is one-dimensional. Then, for every~$k$, the $k$-clustering computed by Ward on~$P$ is an $\mathcal{O}(1)$-approximation with respect to the $k$-means objective function.
\end{theorem}

For the purpose of analyzing the worst-case behavior of Ward's method, an instance sometimes also contains an integer $k \in \mathbb{N}$ in addition to $P$ (even though Ward itself only takes $P$ as the input). If we specify $P$ and $k$, then we are interested in the quality of the $k$-clustering produced by Ward on $P$.

We will usually denote the hierarchical clustering computed by Ward on $P$ by $\mathcal{W}=(\mathcal{W}_0,\ldots,\mathcal{W}_{n-1})$. Ward's method always chooses greedily a cheapest merge to perform. We say that a merge is a \emph{greedy merge} if it is a cheapest merge; if all merges are greedy, we call $\mathcal{W}$ greedy. Ward's method computes a greedy hierarchical clustering, and every greedy hierarchical clustering can be the output of Ward's method.

\subsection{Prelude: Reordering}\ \\
Recall the following statement from §\ref{techniques}:

\convexlemma*
Lemma~\ref{konvex} means that Ward will always merge $A$ and $C$ or $B$ and $C$, and never $A$ and $B$. This gives us a convexity property: If Ward forms a cluster $M$, then no other point or cluster lies within the convex hull of $M$. Clusters can thus also never overlap, and we get a concept of neighbors on the line. Thus, the clusterings $\mathcal{W}_i$ consist of non-overlapping clusters, which we can thus view as ordered by their position on the line. Ward's method always merges neighbors on the line. We will combine it with the following useful corollary of Lemma~\ref{obsb}. It gives a condition under which merging a cluster $A$ with a subcluster $B' \subset B$ is cheaper than merging $A$ with $B$. Notice that without the condition, the statement is not true: Imagine that $A$ and $B$ have the same centroid (merging them is free), but $\mu(B') \neq \mu(B)$. Then clearly, merging $A$ with $B'$ is more expensive than merging $A$ and $B$.

\begin{restatable}{corollary}{subclusterlemma}\label{subcluster}
	Assume we have two finite clusters $B' \subseteq B \subset \mathbb R^d$ and a third finite cluster $A \subset \mathbb R^d$ such that $||\mu(A)-\mu(B')||^2 \leq ||\mu(A)-\mu(B)||^2$. Then $D(A,B')\leq D(A,B)$.
\end{restatable}

Corollary~\ref{subcluster} holds in arbitrary dimension. However, for $d=1$, it is much easier to benefit from it. We get a very convenient tool that we call \emph{reordering}. Say that Ward at some point merges two clusters $A$ and $B$. By Lemma~\ref{konvex}, that means that $\mu(A)$ and $\mu(B)$ are neighbors on the line (at the time of the merge). Now assume that $A$ and $B$ are present for a while before they are merged. Then during all this time, they are neighbors. Notice that this means that merging $A$ and $B$ will result in a centroid $\mu(A\cup B)$ which is further away from any other cluster than $\mu(A)$ and $\mu(B)$ are. So, clusters that did not want to merge with $A$ or $B$ would also not merge with $A\cup B$ by Corollary~\ref{subcluster}. Thus, we could perform the merge $(A,B)$ \emph{earlier} without distorting Ward's course of action at all (except that the merge $(A,B)$ is at the wrong position). Lemma~\ref{lemma:reorderingward} below formulates this idea.

Recall that a hierarchical clustering can also be described by the $n-1$ merge operations that produce it. We usually denote the sequence of merges by $(A,B)(\mathcal{W})=((A_1,B_1),\ldots,(A_{n-1},B_{n-1}))$. We say that a cluster $Q \subset P$ \emph{exists} in $\mathcal{W}$ after merge $t$ if $Q \in \mathcal{W}_t$. If $Q$ is the result of the merge $(A_i,B_i)$ (i.e., $Q=A_i\cup B_i)$, and it is later merged with another cluster in merge $(A_j,B_j)$ (i.e., $A_j=Q$ or $B_j=Q$), then $Q$ exists as long as merge $i$ has happened and merge $j$ has not yet happened. All singleton clusters exist in $\mathcal{W}_0$. After merge $n-1$, $P$ is the only remaining existing cluster. 

\begin{lemma}[Reordering Lemma]\label{lemma:reorderingward}
Let $P\subset\mathcal{R}^d$ be an input for which Ward computes the clustering $\mathcal{W}$ with merge operations $(A,B)(\mathcal{W})$. Consider the merge $(A_t,B_t)$ for $t \in [n-1]$.
If both $A_t$ and $B_t$ exist after merge $s < t$, then
\begin{enumerate}
\item\label{lemma:reorderingward:1} The sequence of merge operations 
\begin{align*}
(A',B')=&(A_1,B_1), \ldots, (A_s,B_s),(A_t,B_t),\\
&(A_{s+1},B_{s+1}),\ldots, (A_{t-1},B_{t-1}),\\
&(A_{t+1},B_{t+1}), \ldots, (A_{n-1},B_{n-1})
\end{align*}
results in a valid hierarchical clustering $\mathcal{W}'$.
\item\label{lemma:reorderingward:1b} $\mathcal{W}'_j = \mathcal{W}_j$ for all $j \ge t$. 
\item\label{lemma:reorderingward:2} All merges except the moved merge $(A'_{s+1},B'_{s+1})=(A_t,B_t)$ are greedy merges. 
\end{enumerate}
\end{lemma}
\begin{proof}
\eqref{lemma:reorderingward:1} and \eqref{lemma:reorderingward:1b} hold because performing merges in a different order does not change the resulting clustering, and after merge $t$, all deviations from the original order are done. For \eqref{lemma:reorderingward:2}, we have to argue that inserting $(A_t,B_t)$ as step $s+1$ does not create cheaper merges. For this, we observe that by Lemma~\ref{konvex}, $A_t$ and $B_t$ are neighbors on the line. 
In the original sequence, no cluster was merged with $A_t$ or $B_t$ up to point $t$. The cluster $A_t\cup B_t$ is a superset of $A_t$ and of $B_t$, and its centroid is further away from all other clusters than the centroids of $A_t$ and $B_t$. Thus by Corollary~\ref{subcluster}, up to point $t$, merging with $A_t\cup B_t$ cannot be cheaper than the merges we do. However, after $(A_{t-1},B_{t-1})$, the clustering is identical to $\mathcal{W}_{t}$ by \eqref{lemma:reorderingward:1}, thus all remaining merges are also greedy merges. 
\end{proof}
Lemma~\ref{lemma:reorderingward} a crucial observation to allow us to systematically analyze Ward's steps: We can sort them into steps that depend on each other, and then analyze them in batches / phases. 

In $\mathbb{R}^d$ for $d > 1$, reordering does not work. Also, we cannot assume that there are no inner-cluster merges. 

\subsection{Prelude: No Inner-cluster Merges}

Reordering also gives us a nice simplification tool. Assume that $A$ and $B$ are in fact singleton clusters, $A=\{a\}$ and $B=\{b\}$, and they are from the same optimum cluster. Then they are present from the start; we can reorder the merge $(A,B)$ to be the first merge Ward does. Indeed, instead of actually doing this merge, we can also simply forget about it and replace $a$ and $b$ by a weighted point. How does this affect the approximation ratio? Both Ward's cost and the optimal cost decrease by $\Delta(\{a,b\})$, meaning that the approximation ratio can only get worse. We can now assume that there are no merges between inner clusters, since inner clusters arise from merging input points that belong to the same optimum cluster.
We formalize our observation in Lemma~\ref{lem:no-in-optcluster-merging}. 

We directly apply Lemma~\ref{lemma:reorderingward} in order to achieve a simplification method. Recall that (given an optimal $k$-clustering) we call a merge $(A_i,B_i)$ an inner-cluster merge if $A_i$ and $B_i$ are inner clusters from the same optimum cluster. For a worst-case instance $(P,k)$ we can always assume that such inner-cluster merges do not happen, as they are only helpful for Ward's method. We formally see this in the next lemma, where we relocate inner-cluster merges to the front of the hierarchical clustering and then eliminate them.

Recall that $\Delta_k(\mathcal{W})=\sum_{Q\in\mathcal{W}_{n-k}}\Delta(Q)$ is the cost of the $k$-clustering contained in $\mathcal{W}$. For an instance $(P,k)$ and Ward's resulting clustering $\mathcal{W}$, the approximation ratio of Ward's method is $\Delta_k(\mathcal{W}) / \opt_k(P)$.

\begin{lemma}\label{lem:no-in-optcluster-merging}
Let $(P,k)$ be an instance with $P\subset\mathbb{R}^d$ and $k\in \mathbb{N}$, for which  $\mathcal{O}=\{O_1,\ldots,O_k\}$ is an optimal $k$-clustering and for which Ward computes the hierarchical clustering $\mathcal{W}$ with merge operations $(A,B)(\mathcal{W})$. Then there exists a weighted point set $P'$ and a hierarchical clustering $\mathcal{W}'$ for $P'$ with merges $(A',B')(\mathcal{W}')$ with the following properties:
\begin{enumerate}
\item\label{lemma:noinnerclustermerges:1} $\mathcal{W}'$ is greedy.
\item\label{lemma:noinnerclustermerges:2} No $(A'_i,B_i')$ is an inner-cluster merge with respect to $\mathcal{O}$.
\item\label{lemma:noinnerclustermerges:3} For some $\alpha \ge 0$, $\Delta_k(\mathcal{W}') = \Delta_k(\mathcal{W}) - \alpha$ and $\opt_k(P')\le\opt_k(P)-\alpha$. 
\end{enumerate}
\end{lemma}
\begin{proof}
Assume that $P$ is weighted; this will be necessary to iterate the following process.
Let $(\{x\},\{y\})$ be a merge operation in $(A,B)(\mathcal{W})$ that merges two points $x,y \in O_j$ for $j\in[k]$, i.e., two points from the same cluster in the optimal solution. Let their weights be $w(x)$ and $w(y)$. By Lemma~\ref{lemma:reorderingward}, we can move the merge $(\{x\},\{y\})$ to the front. 
Then we replace $x$ and $y$ in $P$ by one point $z=\frac{w(x) x + w(y) y}{w(x)+w(y)}$ with weight $w(z) := w(x)+w(y)$. By Lemma~\ref{obsb}, $z$ behaves identically to $\{x,y\}$ in Ward's method. Thus, we can adjust ${\mathcal{W}}'$ by removing the merge operation $(\{x\},\{y\})$, and replacing $x$ and $y$ by $z$ in all further merge operations of the cluster $\{x,y\}$. We see that~\eqref{lemma:noinnerclustermerges:1} holds for the new hierarchical clustering. 
Our adjustment will change the cost by $\alpha:=\Delta(\{x,y\})$.
Similarly, we can replace $x$ and $y$ in $O_j$ by $z$, which decreases the cost of the clustering induced by $O_1,\ldots,O_k$ by $\alpha$. Since this is still a possible clustering, the optimal clustering can cost at most $\opt_k(P)-\alpha$. 
Thus, \eqref{lemma:noinnerclustermerges:3} holds for the new clustering.

Observe that if~\eqref{lemma:noinnerclustermerges:2} is not true, then there has to be a merge operation where two points from the same cluster in the optimum are merged. 
Thus, we can complete the proof by repeating the above process until we have removed all pairs with this property. Then~\eqref{lemma:noinnerclustermerges:2} holds.
\end{proof}

Now if Ward performs inner-cluster merges on an instance, we apply  Lemma~\ref{lem:no-in-optcluster-merging}. If this changes the optimum solution, we just apply Lemma~\ref{lem:no-in-optcluster-merging} again, and repeat this until Ward does not do any inner-cluster merges.
We explicitly note the following trivial corollary.

\begin{corollary}\label{corollary:innerarepoints}
Assume that $\mathcal{W}'$ and $(A',B')(\mathcal{W'})$ result from applying Lemma~\ref{lem:no-in-optcluster-merging} until Ward does not do inner cluster merges. If a merge $(A_i',B_i')$ for $i \in [n-1]$ contains an inner cluster, then this inner cluster is a (weighted) input point.
\end{corollary}
\begin{proof}
If $A$ resulted from a previous merge, then this merge was an inner-cluster merge, which is a contradiction.
\end{proof}

Corollary~\ref{corollary:innerarepoints} implies that we can use the terms inner cluster and input point interchangeably.

\subsection{Prelude: Clustering points together}
Crucial in showing the approximation factors of the good merges is the following lemma. To see its usage, assume that $A$ and $B$ belong to one optimum cluster, and $C$ and $D$ belong to another. Then the lemma implies that if Ward has already merged $B$ and $C$, but $\Delta(B\cup C)$ is small, say $\Delta(B\cup C) \le c\cdot(\Delta(B)+\Delta(C))$, then we can still obtain a $7c$-approximation. Its proof is deferred to the full version of this paper.

\begin{restatable}{lemma}{goodmergefour}\label{lem:goodmergewithfour}
Let $A,B,C,D\subset\mathbb{R}^d$ be disjoint sets with $|A| \le |B|$ and $|C|\ge |D|$. Then
\[
\begin{split}
&\Delta(A\cup B\cup C\cup D) \le \Delta(A) + 3 \cdot \Delta(B\cup C)\\
& \hspace{1.5cm} + \Delta(D) + 4 \cdot D(A,B) + 4\cdot D(C,D)
\end{split}
\]
and
\[
\begin{split}
&D(A\cup B, C \cup D) \le 3 \cdot \Delta(B\cup C) + 3 \cdot D(A,B)\\
& \hspace{1.5cm}+ 3 \cdot D(C,D) - \Delta(B) - \Delta(C).
\end{split}
\]
\end{restatable}

\subsection{The analysis}

We now analyze the worst-case behavior of Ward's method on the line by fixing an arbitrary worst-case example that does not contain inner-cluster merges (we can assume this by Lemma~\ref{lem:no-in-optcluster-merging}).

The general plan is the following. Whenever Ward merges two clusters, it does so greedily, meaning that the cost of the merge is always bounded by the cost of any other merge. Thus, if we can find a merge with low cost, then the merge actually performed can only be cheaper. We can clearly find cheap merges in the beginning, however, Ward's decisions may lead us to a situation where we run out of the originally good options. The idea of the proof is to find a point during Ward's execution where:
\begin{itemize}
\item We still know a bound on the costs produced so far.
\item We know a set $\mathcal{S}$ of good merges that can still be performed and lead to a good $k$-clustering.
\item We can ensure that no merge can possibly destroy two merges from $\mathcal{S}$.
\end{itemize}
At such a point in time, we can use $\mathcal{S}$ to charge the remaining merges that Ward does to compute a $k$-clustering. We find this point in time by sorting specific merges of Ward into the front, and bounding their cost. There will be five phases of merges which we need to pull forward and charge.

\paragraph*{The phases}
We will use the reordering lemma (Lemma~\ref{lemma:reorderingward}) to sort the merges into phases and then analyze the cost of the solution after each phase. 

In the following, we call a cluster that contains points from more than one optimum cluster \emph{composed}, more precisely, we call it an \emph{$\ell$-composed cluster} if it contains points from $\ell$ different optimum clusters. Most of the time, we are interested in $2$-composed clusters, and we name such a cluster \emph{$2$-composed cluster from $O_j$ and $O_{j+1}$} if these are the involved optimum clusters. 

\begin{figure*}
\centering
\begin{tikzpicture}[xscale=1.5,yscale=1.1]
\draw [ultra thick] (1,-0.5) -- (1,1);

\draw [thick,black!40] (-1.5,0.25) -- (3.5,0.25);

\node at (0,1.2) {$O_{j}$};
\node at (2.25,1.2) {$O_{j+1}$};

\node [circle,minimum width=0.2cm,inner sep=0cm,fill=black,label=below:{$x_{\ell}$}] at (-0.75,0.25) {};
\node [circle,minimum width=0.2cm,inner sep=0cm,fill=black] at (-0.4,0.25) {};
\node [circle,minimum width=0.2cm,inner sep=0cm,fill=black] at (0.1,0.25) {};
\node [circle,minimum width=0.2cm,inner sep=0cm,fill=black] at (0.5,0.25) {};
\node [circle,minimum width=0.2cm,inner sep=0cm,fill=black] at (0.75,0.25) {};
\node [circle,minimum width=0.2cm,inner sep=0cm,fill=black] at (1.3,0.25) {};
\node [circle,minimum width=0.2cm,inner sep=0cm,fill=black] at (1.8,0.25) {};
\node [circle,minimum width=0.2cm,inner sep=0cm,fill=black] at (2.1,0.25) {};
\node [circle,minimum width=0.2cm,inner sep=0cm,fill=black] at (2.3,0.25) {};
\node [circle,minimum width=0.2cm,inner sep=0cm,fill=black] at (2.5,0.25) {};
\node [circle,minimum width=0.2cm,inner sep=0cm,fill=black,label=below:{$x_r$}] at (2.8,0.25) {};

\begin{scope}[xshift=5.5cm]
\node at (1,-0.8) {Creation: Phase $P1$};
\draw [rounded corners,fill=black!10,draw=none] (0.6,0) rectangle (1.45,0.5);

\draw [ultra thick] (1,-0.5) -- (1,1);

\draw [thick,black!40] (-1.5,0.25) -- (3.5,0.25);

\node at (0,1.2) {$O_{j}$};
\node at (2.25,1.2) {$O_{j+1}$};

\node [circle,minimum width=0.2cm,inner sep=0cm,fill=black,label=below:{$x_{\ell}$}] at (-0.75,0.25) {};
\node [circle,minimum width=0.2cm,inner sep=0cm,fill=black] at (-0.4,0.25) {};
\node [circle,minimum width=0.2cm,inner sep=0cm,fill=black] at (0.1,0.25) {};
\node [circle,minimum width=0.2cm,inner sep=0cm,fill=black] at (0.5,0.25) {};
\node [circle,minimum width=0.2cm,inner sep=0cm,fill=black] at (0.75,0.25) {};
\node [circle,minimum width=0.2cm,inner sep=0cm,fill=black] at (1.3,0.25) {};
\node [circle,minimum width=0.2cm,inner sep=0cm,fill=black] at (1.8,0.25) {};
\node [circle,minimum width=0.2cm,inner sep=0cm,fill=black] at (2.1,0.25) {};
\node [circle,minimum width=0.2cm,inner sep=0cm,fill=black] at (2.3,0.25) {};
\node [circle,minimum width=0.2cm,inner sep=0cm,fill=black] at (2.5,0.25) {};
\node [circle,minimum width=0.2cm,inner sep=0cm,fill=black,label=below:{$x_r$}] at (2.8,0.25) {};
\end{scope}

\begin{scope}[yshift=-3cm]
\node at (1,-0.8) {Growth: Phase $P2$};
\draw [rounded corners,fill=black!10,draw=none] (-0.6,-0.15) rectangle (1.95,0.65);

\draw [ultra thick] (1,-0.5) -- (1,1);

\draw [thick,black!40] (-1.5,0.25) -- (3.5,0.25);

\node at (0,1.2) {$O_{j}$};
\node at (2.25,1.2) {$O_{j+1}$};

\node [circle,minimum width=0.2cm,inner sep=0cm,fill=black,label=below:{$x_{\ell}$}] at (-0.75,0.25) {};
\node [circle,minimum width=0.2cm,inner sep=0cm,fill=black] at (-0.4,0.25) {};
\node [circle,minimum width=0.2cm,inner sep=0cm,fill=black] at (0.1,0.25) {};
\node [circle,minimum width=0.2cm,inner sep=0cm,fill=black] at (0.5,0.25) {};
\node [circle,minimum width=0.2cm,inner sep=0cm,fill=black] at (0.75,0.25) {};
\node [circle,minimum width=0.2cm,inner sep=0cm,fill=black] at (1.3,0.25) {};
\node [circle,minimum width=0.2cm,inner sep=0cm,fill=black] at (1.8,0.25) {};
\node [circle,minimum width=0.2cm,inner sep=0cm,fill=black] at (2.1,0.25) {};
\node [circle,minimum width=0.2cm,inner sep=0cm,fill=black] at (2.3,0.25) {};
\node [circle,minimum width=0.2cm,inner sep=0cm,fill=black] at (2.5,0.25) {};
\node [circle,minimum width=0.2cm,inner sep=0cm,fill=black,label=below:{$x_r$}] at (2.8,0.25) {};
\end{scope}

\begin{scope}[yshift=-3cm,xshift=5.5cm]
\node at (1,-0.8) {Left side done: Phase $P3$};
\draw [rounded corners,fill=black!10,draw=none] (-1,-0.25) rectangle (1.95,0.75);

\draw [ultra thick] (1,-0.5) -- (1,1);

\draw [thick,black!40] (-1.5,0.25) -- (3.5,0.25);

\node at (0,1.2) {$O_{j}$};
\node at (2.25,1.2) {$O_{j+1}$};

\node [circle,minimum width=0.2cm,inner sep=0cm,fill=black,label=below:{$x_{\ell}$}] at (-0.75,0.25) {};
\node [circle,minimum width=0.2cm,inner sep=0cm,fill=black] at (-0.4,0.25) {};
\node [circle,minimum width=0.2cm,inner sep=0cm,fill=black] at (0.1,0.25) {};
\node [circle,minimum width=0.2cm,inner sep=0cm,fill=black] at (0.5,0.25) {};
\node [circle,minimum width=0.2cm,inner sep=0cm,fill=black] at (0.75,0.25) {};
\node [circle,minimum width=0.2cm,inner sep=0cm,fill=black] at (1.3,0.25) {};
\node [circle,minimum width=0.2cm,inner sep=0cm,fill=black] at (1.8,0.25) {};
\node [circle,minimum width=0.2cm,inner sep=0cm,fill=black] at (2.1,0.25) {};
\node [circle,minimum width=0.2cm,inner sep=0cm,fill=black] at (2.3,0.25) {};
\node [circle,minimum width=0.2cm,inner sep=0cm,fill=black] at (2.5,0.25) {};
\node [circle,minimum width=0.2cm,inner sep=0cm,fill=black,label=below:{$x_r$}] at (2.8,0.25) {};
\end{scope}

\begin{scope}[yshift=-6cm]
\node at (1,-0.8) {Growth: Phase $P4$};
\draw [rounded corners,fill=black!10,draw=none] (-1,-0.25) rectangle (2.65,0.75);

\draw [ultra thick] (1,-0.5) -- (1,1);

\draw [thick,black!40] (-1.5,0.25) -- (3.5,0.25);

\node at (0,1.2) {$O_{j}$};
\node at (2.25,1.2) {$O_{j+1}$};

\node [circle,minimum width=0.2cm,inner sep=0cm,fill=black,label=below:{$x_{\ell}$}] at (-0.75,0.25) {};
\node [circle,minimum width=0.2cm,inner sep=0cm,fill=black] at (-0.4,0.25) {};
\node [circle,minimum width=0.2cm,inner sep=0cm,fill=black] at (0.1,0.25) {};
\node [circle,minimum width=0.2cm,inner sep=0cm,fill=black] at (0.5,0.25) {};
\node [circle,minimum width=0.2cm,inner sep=0cm,fill=black] at (0.75,0.25) {};
\node [circle,minimum width=0.2cm,inner sep=0cm,fill=black] at (1.3,0.25) {};
\node [circle,minimum width=0.2cm,inner sep=0cm,fill=black] at (1.8,0.25) {};
\node [circle,minimum width=0.2cm,inner sep=0cm,fill=black] at (2.1,0.25) {};
\node [circle,minimum width=0.2cm,inner sep=0cm,fill=black] at (2.3,0.25) {};
\node [circle,minimum width=0.2cm,inner sep=0cm,fill=black] at (2.5,0.25) {};
\node [circle,minimum width=0.2cm,inner sep=0cm,fill=black,label=below:{$x_r$}] at (2.8,0.25) {};
\end{scope}

\begin{scope}[yshift=-6cm,xshift=5.5cm]
\node at (1,-0.8) {Both sides done: (Phase $P5$)};
\draw [rounded corners,fill=black!10,draw=none] (-1,-0.25) rectangle (3,0.75);

\draw [ultra thick] (1,-0.5) -- (1,1);

\draw [thick,black!40] (-1.5,0.25) -- (3.5,0.25);

\node at (0,1.2) {$O_{j}$};
\node at (2.25,1.2) {$O_{j+1}$};

\node [circle,minimum width=0.2cm,inner sep=0cm,fill=black,label=below:{$x_{\ell}$}] at (-0.75,0.25) {};
\node [circle,minimum width=0.2cm,inner sep=0cm,fill=black] at (-0.4,0.25) {};
\node [circle,minimum width=0.2cm,inner sep=0cm,fill=black] at (0.1,0.25) {};
\node [circle,minimum width=0.2cm,inner sep=0cm,fill=black] at (0.5,0.25) {};
\node [circle,minimum width=0.2cm,inner sep=0cm,fill=black] at (0.75,0.25) {};
\node [circle,minimum width=0.2cm,inner sep=0cm,fill=black] at (1.3,0.25) {};
\node [circle,minimum width=0.2cm,inner sep=0cm,fill=black] at (1.8,0.25) {};
\node [circle,minimum width=0.2cm,inner sep=0cm,fill=black] at (2.1,0.25) {};
\node [circle,minimum width=0.2cm,inner sep=0cm,fill=black] at (2.3,0.25) {};
\node [circle,minimum width=0.2cm,inner sep=0cm,fill=black] at (2.5,0.25) {};
\node [circle,minimum width=0.2cm,inner sep=0cm,fill=black,label=below:{$x_r$}] at (2.8,0.25) {};
\end{scope}
\end{tikzpicture}
\caption{The pricipal phases of development of a $2$-composed cluster.\label{figure:phaseexample}}
\end{figure*}
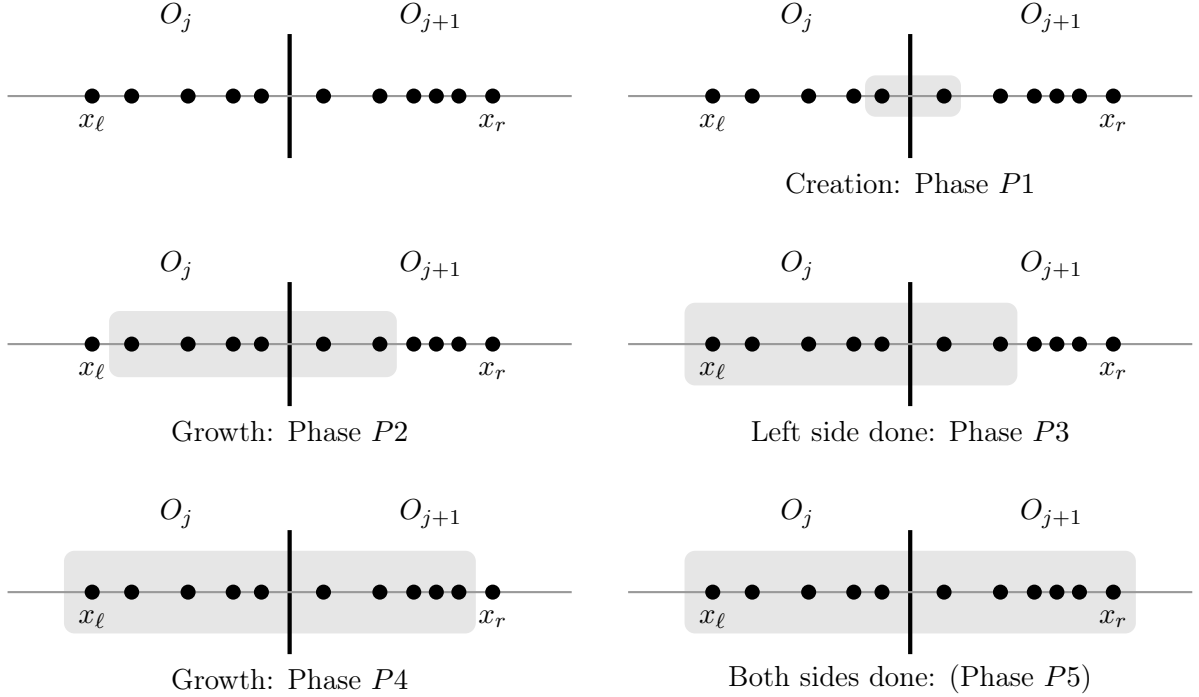
The goal of the reordering is simple in nature; we want to collect all merges that create $2$-composed clusters and that grow $2$-composed clusters. We can think of the phases as different stages of development of $2$-composed clusters. A $2$-composed cluster may become part of the $k$-clustering computed by Ward's method, or it may at some point become $i$-composed for $i>2$, at which time we are no longer interested in it. By the \emph{final stage} of a $2$-composed cluster we either mean how it looks in the $k$-clustering, or how it looked in the last step before it became more than $2$-composed.

Consider the example in Figure~\ref{figure:phaseexample}, where we depict the development of a $2$-composed cluster from $O_j$ and $O_{j+1}$ which in its final stage consists of the input points $x_\ell,\ldots,x_r$. It undergoes five principal phases: It is created by merging a point from $O_j$ with a point from $O_{j+1}$ (phase $P1$). Then it grows; it is merged with points left and right of itself (phase $P2$). We add extra phases for the last points on both sides. In phase $P3$, the first side is completed; in the example, it is the left side. This merge is again followed by a growth phase (phase $P4$). The final phase $P5$ consists of the final merge on the other side; the right side in the example. (We skip some merges in $P5$,  the details of $P5$ are not discussed until much later in this proof).

So, we use reordering to pull the following phases of merges to the front. 

\begin{enumerate}
\item[P1] (Creation phase)\\
We create $2$-composed clusters by collecting the  merges $(\{a_i\},\{b_i\})$ with $a_i \in O_j$, $b_i \in O_{j+1}$ for some $j \in [k]$. 
The collected merges constitute phase $P1$. 
For technical reasons, we make one exception. If the $2$-composed cluster only consists of two input points in its final stage (i.e., the creating merge is also the last merge), then we defer the merge to phase $P5$. \item[P2] (Main growth phase)\\
We now grow the $2$-composed clusters initialized during phase $P1$.
For each $2$-composed cluster, we move the growth merges to phase $P2$, preserving their original order. We stop right before one side of the $2$-composed cluster is done. There may be many growth merges for a cluster, or none.
\item[P3] (First side elimination phase)\\
This phase consists of at most one merge for each $2$-composed cluster, and this merge is the last merge on the first side. After phase $P3$, every $2$-composed cluster thus has one side where it will not be merged with further input points. Notice that a cluster may skip phase $P3$ if it only shares one point with $O_j$ or $O_{j+1}$ in its final stage, anyway. 
\item[P4] (Second growth phase)\\
This phase resembles phase $P2$, however, the growth is now one-sided. For each $2$-composed cluster, we move the growth merges to phase $P4$, preserving their original order, and stopping right before the final merge.
\item[P5] (Second side elimination phase)\\
The last phase consists of at most one merge for each cluster. If the final stage of a $2$-composed cluster contains only two points, then the merging of these two points is done in phase $P5$. Otherwise, phase $P5$ may contain the last merge for the cluster, resulting in its final state. For technical reasons, we have to exclude some merges; we postpone the details to Definition~\ref{def:phase5}.
\end{enumerate}

We now analyze the sum of the $1$-means costs of all clusters in the clustering after each phase. 
The proofs of the lemmata are deferred to the full version of the paper.
We start with phases $P1$ and $P2$.

\begin{lemma}\label{lem:phase12}
Let $N=\{x_{a},\ldots,x_{b}\}$ with $x_{a},\ldots,x_m \in O_j$ and $x_{m+1},\ldots,x_{b} \in O_{j+1}$ be a $2$-composed cluster after phases $P1$ and $P2$. Then
\[
\Delta(N) \le \sum_{h=a-1}^{m-1} D(x_h,x_{h+1}) + \sum_{h=m+1}^{b} D(x_h,x_{h+1}).
\]
Furthermore, $D(N\cap O_j,N\cap O_{j+1}) \le D(x_{a-1},x_{a}) + D(x_{b},x_{b+1})$.
\end{lemma}

In phase $P3$, Ward's method faces the first situation where it may run out of good merge options and has to resort to more expensive merges. Notice that by the definition of our phases, each cluster has one side where after phase $P2$, there is exactly one point left which has not been added to the cluster. 
The key technical observation that we use again and again during the (omitted) proofs is the following corollary. 

\goodmergethree*

We need the following interpretation of Corollary~\ref{cor:goodmergewiththree}. If we have a 2-composed cluster $M=A\cup B$ which consists of a lighter cluster $A \subseteq O'$ for an optimum cluster $O'$ and a heavier cluster $B \subset O''$ for another optimum cluster $O''$, then merging $A\cup B$ with another cluster $C \subset O''$ basically costs as much as $A\subseteq O'$ and $B\cup C \subseteq O''$ cost individually, plus what merging $A$ and $B$ costed us already (up to constant factors). Corollary~\ref{cor:goodmergewiththree} allows us to analyze the $1$-means costs of the clusters after phase $P4$.

\begin{lemma}\label{lem:phase34}
Let $F=\{x_{\ell},\ldots,x_r\}$ be the final state of a $2$-composed cluster, with $x_{\ell},\ldots,x_m \in O_j$ and $x_{m+1},\ldots,x_r \in O_{j+1}$. 
The state of the cluster after phase $P4$ is either $N=\{x_{\ell},\ldots,x_{r-1}\}$ or $N=\{x_{\ell-1},\ldots,x_r\}$. In both cases,
\[
\Delta(N) \le 8 \cdot(\Delta(\{x_{\ell},\ldots,x_m\}) + \Delta(\{x_{m+1},\ldots,x_{r}\})).
\]
\end{lemma}

Now we come to phase $P5$, which we haven't completely defined yet. The problem with phase $P5$ is that we can no longer charge all clusters \lq internally\rq. To see this, first notice that we say that a $2$-composed cluster $F$ from $O_j$ and $O_{j+1}$ \emph{points to cluster $A$} if 
\begin{itemize}
\item $w(F\cap O_j) \ge w(F\cap O_{j+1})$ holds and $A$ is the cluster left of $F$, or
\item $w(F\cap O_j) \le w(F\cap O_{j+1})$ holds and $A$ is the cluster right of $F$.
\end{itemize}
We define a \emph{lopsided cluster} to be a $2$-composed cluster $F = \{x_\ell,\ldots,x_r\}$ for which the last merge is $\{F\backslash\{x\},\{x\}\}$, but at the time of this merge, $F'=F\backslash\{x\}$ does not point to $\{x\}$. This means that we cannot use Corollary~\ref{cor:goodmergewiththree} (directly) to charge this merge. As a technicality, we also call a $2$-composed cluster lopsided if it only contains two points in its final state; again, we cannot use Corollary~\ref{cor:goodmergewiththree} in this case. 

We have to pay attention to one more detail when defining phase $P5$. When charging $2$-composed clusters internally, we could always be sure that the clusters that are involved are part of one of the two optimum clusters that the $2$-composed cluster intersects. That is because the $2$-composed cluster by definition only contains points from two optimum clusters, and we only dealt with points and subclusters of such a $2$-composed cluster. However, in the following arguments, we will have to argue about clusters neighboring a $2$-composed cluster. These may or may not belong to one of the optimum clusters. Let $A$ and $B$ be two clusters that are neighbors on the line such that $A$ lies left of $B$. 
We say that there is an \emph{opt change between $A$ and $B$} if the last point in $A$ and the first point in $B$ belong to different optimum clusters. 

Now we define phase $P5$. Let $Y$ be the cluster that lies on the other side of $F'$ than~$x$ \emph{at the time of the merge $\{F',\{x\}\}$}. Let $Z$ be the cluster that lies \lq behind\rq\ $x$ from the point of view of $F'$ \emph{at the time of the merge $\{F',\{x\}\}$}. By \emph{behind from $F$'s point of view} we mean that if $x$ lies left of $F$, then $Z$ lies left of $x$, and if $x$ lies right of $F'$, then $Z$ lies right of $x$.

\begin{definition}[Phase P5]\label{def:phase5}
Phase $P5$ contains the final merge $\{F',\{x\}\}$ of a cluster $F=F'\cup\{x\}$ if any of the following conditions applies.
\begin{enumerate}
\item $F$ is not lopsided (phase $P5a$),
\item $F$ is lopsided, there is no opt change between $Y$ and $F'$, and $Y$ is an inner cluster (phase $P5b$), 
\item $F$ is lopsided, there is no opt change between $\{x\}$ and $Z$, and  $Z$ is an inner cluster (phase $P5c$),
\item $F$ is lopsided, there is no opt change between $\{x\}$ and $Z$, $Z$ is $2$-composed, and points to \{x\} (phase $P5d$).
\end{enumerate}
\end{definition}

The next lemma deals with merges in $P5a$.

\begin{lemma}\label{lem:phase5a}
Let $F=\{x_{\ell},\ldots,x_r\}$ be the final state of a $2$-composed cluster, with $x_{\ell},\ldots,x_m \in O_j$ and $x_{m+1},\ldots,x_r \in O_{j+1}$. Assume that $F$ is not lopsided.
Then
\[
\Delta(F) \le 35  \cdot(\Delta(\{x_{\ell},\ldots,x_m\}) + \Delta(\{x_{m+1},\ldots,x_{r}\})).
\]
\end{lemma}

Now we consider the merges in phase $P5b$.

\begin{lemma}\label{lem:phase5b}
Let $F=\{x_{\ell},\ldots,x_r\}$ be the final state of a $2$-composed cluster, with $x_{\ell},\ldots,x_m \in O_j$ and $x_{m+1},\ldots,x_r \in O_{j+1}$. Assume that $F$ is lopsided.
Assume that at the time of the merge $\{F\backslash\{x\},\{x\}\}$, the cluster on the other side of $F'=F\backslash \{x\}$ is an inner cluster $Y$, and there is no opt change between $F'$ and $Y$.
Then if $x=x_{\ell}$, we have
\[
\Delta(F) \le 35  \cdot(\Delta(\{x_{\ell},\ldots,x_m\}) + \Delta(\{x_{m+1},\ldots,x_{r+1}\})),
\]
and if $x=x_{r}$, then
\[
\Delta(F) \le 35  \cdot(\Delta(\{x_{\ell-1},\ldots,x_m\}) + \Delta(\{x_{m+1},\ldots,x_{r}\})).
\]
\end{lemma}

The following lemma is the main lemma about the phases and summarizes our findings: After phase $5$, the error is still bounded by a constant times the optimum value.

\begin{lemma}
Let $C_5$ be the clustering after phase $P5$. Then
\[
\sum_{A \in C_5} \Delta(A) \le \mathcal{O}(1) \cdot \opt_k.
\]
\end{lemma}

\paragraph*{Good merges for the final analysis}
In general, the clustering of Ward after phase~$P5$ has still more than~$k$ clusters. It remains to analyze the merges after phase~$P5$ that reduce the number of clusters to~$k$. For the final charging argument, we need four types of \emph{good merges}. Good merges are not necessarily merges that Ward's method does, instead, it's a collection of merges that are possible and can be used for charging. Indeed, good merges include merges that would not be present anymore if Ward did them, since then we would move them to the phases. But if Ward never uses them, they may still be present for us to charge against.

The whole point of the phases is to ensure that any merge that Ward may still do does not destroy two good merges. The final arguments of the proof will be to count good merges and to show that no two good merges can be invalidated simultaneously by one of Ward's merges.

Recall that $W_1,\ldots,W_{\ell}$ is the current Ward solution, and $O_1,\ldots,O_k$ is a fixed optimal solution, numbered from left to right. The following merges are good merges in the sense that we can bound the increase in cost. Of course, the result of the merge only forms a cluster of low cost if the participating clusters had low cost beforehand.

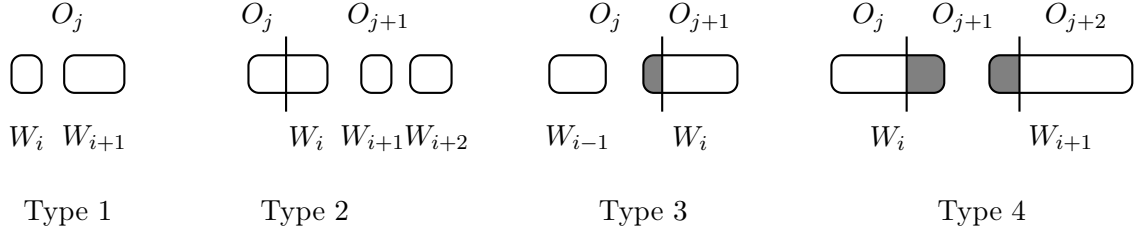
\begin{figure*}
\centering
\begin{tikzpicture}

\begin{scope}

\node at (0.8,-0.6) {$W_{i}$};
\draw [thick,rounded corners] (0.6,0) rectangle (1.0,0.5);

\node at (1.7,-0.6) {$W_{i+1}$};
\draw [thick,rounded corners] (1.3,0) rectangle (2.1,0.5);

\node at (1.35,1) {$O_{j}$};

\node at (1.35,-1.6) {Type $1$};
\end{scope}

\begin{scope}[xshift=3.25cm]

\node at (1.275,-0.6) {$W_{i}$};
\draw [thick,rounded corners] (0.5,0) rectangle (1.55,0.5);
\draw [thick] (1,-0.25) -- (1,0.75);

\node at (2.15,-0.6) {$W_{i+1}$};
\draw [thick,rounded corners] (2,0) rectangle (2.4,0.5);

\node at (3,-0.6) {$W_{i+2}$};
\draw [thick,rounded corners] (2.65,0) rectangle (3.2,0.5);

\node at (0.625,1) {$O_{j}$};
\node at (2.225,1) {$O_{j+1}$};

\node at (1.25,-1.6) {Type $2$};
\end{scope}

\begin{scope}[xshift=7.75cm]
\node at (0.375,-0.6) {$W_{i-1}$};
\draw [thick,rounded corners] (0,0) rectangle (0.75,0.5);

\node at (1.875,-0.6) {$W_{i}$};
\begin{scope}
  \clip (1,0) rectangle (1.5,0.5);
  \fill[gray,rounded corners] (1.25,0) rectangle (2.5,0.5);
\end{scope}
\draw [thick,rounded corners] (1.25,0) rectangle (2.5,0.5);
\draw [thick] (1.5,-0.25) -- (1.5,0.75);

\node at (0.75,1) {$O_{j}$};
\node at (2,1) {$O_{j+1}$};

\node at (1.25,-1.6) {Type $3$};
\end{scope}

\begin{scope}[xshift=11.5cm]
\node at (0.75,-0.6) {$W_i$};
\begin{scope}
  \clip (1,0) rectangle (3,0.5);
  \fill[gray,rounded corners] (0,0) rectangle (1.5,0.5);
\end{scope}
\draw [thick,rounded corners] (0,0) rectangle (1.5,0.5);
\draw [thick] (1,-0.25) -- (1,0.75);

\node at (3.05,-0.6) {$W_{i+1}$};
\begin{scope}
  \clip (2,0) rectangle (2.5,0.5);
  \fill[gray,rounded corners] (2.1,0) rectangle (4,0.5);
\end{scope}
\draw [thick,rounded corners] (2.1,0) rectangle (4,0.5);
\draw [thick] (2.5,-0.25) -- (2.5,0.75);

\node at (0.5,1) {$O_j$};
\node at (1.75,1) {$O_{j+1}$};
\node at (3.25,1) {$O_{j+2}$};

\node at (2,-1.6) {Type $4$};
\end{scope}

\end{tikzpicture}
\caption{Different types of good merge situations.\label{fig:goodmerges} A part of a Ward cluster that is filled with gray contains more points than the white part of the same Ward cluster.}
\end{figure*}

\begin{itemize}
\item[\emph{T1}:] Two inner clusters $W_i,W_{i+1}$ of the same optimal cluster $O_j$, i.e., $W_i,W_{i+1} \subset O_j$. This type of merge is never actually applied by Ward on simplified examples, but we need it for charging.
\item[\emph{T2}:] A $2$-composed cluster $W_i \subset O_j \cup O_{j+1}$ for some $j$ and an inner cluster $W_{i+1} \subset O_{j+1}$, with the condition that $W_{i+2}$ is an inner cluster of $O_{j+1}$ as well. Also: The symmetric situation of a $2$-composed cluster $W_i \subset O_j \cup O_{j+1}$ for some $j$ and an inner cluster $W_{i-1} \subset O_{j}$ with the condition that $W_{i-2} \subset O_{j}$.

\item[\emph{T3}:] A $2$-composed cluster $W_i \subset O_j \cup O_{j+1}$ for some $j$ and an inner cluster $W_{i-1} \subset O_j$, with the condition that $W_i$ points to $W_{i-1}$. Also: The symmetric situation of a $2$-composed cluster $W_i \subset O_j \cup O_{j+1}$ for some $j$ and an inner cluster $W_{i+1} \subset O_{j+1}$ with the condition that $W_i$ points to $W_{i+1}$.
\item[\emph{T4}:] Two $2$-composed clusters $W_i \subset O_j \cup O_{j+1}$  and $W_{i+1} \subset O_{j+1} \cup O_{j+2}$ that point at each other.
\end{itemize}

We already know T1 merges (inner-cluster merges), T2 merges (growth phase and phase $5c$) and T3 merges (merges chargeable with Corollary \ref{cor:goodmergewiththree}). We know that applying them increases the cost by at most a constant factor. We also know that these merges cannot happen anymore: T1 merges are inner-cluster merges, which Ward does not do on our example. T2 merges happen either in the growth phase, or in phase $5c$. T3 merges merge non-lopsided clusters, which happens in phase $5a$.

T4 is a type of merge that we did not yet consider, and which Ward can still do. Indeed, to charge it, we need the general charging statement in the below Lemma~\ref{lem:goodmergewithfour} from which Corollary \ref{cor:goodmergewiththree} follows.

\goodmergefour*

Let $W_i$ and $W_{i+1}$ constitute a T4 merge as described above. Then Lemma~\ref{lem:goodmergewithfour} with $A=W_i \cap O_j$, $B=W_i \cap O_{j+1}$, $C=W_{i+1}\cap O_{j+1}$ and $D=W_{i+1}\cap O_{j+2}$ implies that
\begin{align*}
&\Delta(W_i \cup W_{i+1})\\
 \le\ & \Delta(W_i \cap O_{j}) + 3 \Delta(O_{j+1}) + \Delta(W_{i+1}\cap O_{j+2})\\
& \hspace*{2cm} + 4D(W_i\cap O_j,W_{i}\cap O_{j+1})\\
& \hspace*{2cm} + 4D(W_{i+1}\cap O_{j+1},W_{i+1}\cap O_{j+2}).
\end{align*}
Thus, if $\Delta(W_i)+\Delta(W_{i+1})$ was bounded by a constant factor times the optimal cost of the points in $W_i\cup W_{i+1}$, then this is still true after the merge of $W_i$ and $W_{i+1}$ (with a higher factor). 

\subparagraph*{Counting inner clusters}

Observe that the only merges that delete more than one inner cluster are the merges in phase $P1$. All other merges remove either exactly one inner cluster, or none at all. In phase $P2$-$P5$, every merge eliminates exactly one inner cluster. In the beginning, there are $n$ inner clusters.
So if phase $P1$ has $n_1$ merges and $P2$ until $P5$ together have $n_r$ merges, then we have $n - 2n_1 - n_r$ inner clusters after phase $P5$, and we have $n_1$ $2$-composed clusters. The total number of all clusters is $n - n_1 - n_r$. 

Consider the Ward clustering $W_1,\ldots,W_t$ after phase $P5$. We split the clustering into blocks, based on the inner clusters. More precisely, we get $n - 2n_1 - n_r - 1$ blocks that start with an inner cluster, possibly has some $2$-composed clusters and ends with another inner cluster. The blocks overlap in the inner clusters. 

We argue that there is at least one good merge in every block except for $k-n_1-1$ blocks.
The exceptions are the blocks where the optimum cluster changes between start and end, but the change happens between the clusters (not in a $2$-composed cluster). This can only happen $k-n_1-1$ times because $n_1$ of the $k-1$ cluster borders are within $2$-composed clusters. For the remaining blocks, we argue the following.
If there are no $2$-composed clusters in the block, then the two inner clusters are neighbored and form a T1 merge.
If there is only one $2$-composed cluster in the block, then it has to point at an inner cluster and thus there is a T2 or a T3 merge.
If there are multiple $2$-composed clusters, we argue as follows. The first $2$-composed cluster either points left and thus there is a T2 or a T3 merge, or it points to the right. Any further $2$-composed cluster either points to the one before it, forming a T4 merge, or it points to the right. This goes on until we either find a merge, or we find the last $2$-composed cluster, which then has to point at the second inner cluster, forming a T2 or T2 merge.

We collect one good merge from every block and call the resulting set of merges $\mathcal{S}$.
Observe that the cost of all merges in $\mathcal{S}$ together is a constant factor of the cost that we have so far, so all merges together cost $\mathcal{O}(1) \opt_k$.

This argument alone is not enough. The main feature of $\mathcal{S}$ is that every merge that Ward actually performs can make at most one merge from our set invalid. This means that we can charge $n - 2n_1 - n_r  - 1 - (k-n_1-1)$ merges to $\mathcal{S}$. 

Notice that our merges are disjoint except for possible overlap at inner clusters.
Assume that a merge of Ward invalidates two merges from our set. There are two ways how this can happen.
Case one is that Ward's merge is one of the two good merges that are invalidated. Say this merge is called $(A,B)$. Then the second merges involves either $A$ or $B$, say it involves $B$. Thus, there is another cluster $C$ next to $B$, and the merge $(A,B)$ invalidates itself and $(B,C)$. This in particular means that $(A,B)$ is a good merge. Since Ward does not do inner-cluster merges, either $A$ or $B$ has to be $2$-composed, since $(A,B)$ is a merge of Ward. If they are both $2$-composed clusters, then $A$ and $B$ are in the same block, thus $(A,B)$ and $(B,C)$ cannot both be in $\mathcal{S}$. Thus, one is $2$-composed and the other is an inner cluster, i.e., they form a T2 or T3 merge, since $(A,B)$ is supposed to be a good merge. If it is a T3 merge, then $(A,B)$ is not lopsided, and would have happened in phase $P5a$. If it is a T2 merge, then it  is either not lopsided (phase $P5a$), or it is lopsided, but has an inner cluster behind its inner cluster (phase $P5c)$. We conclude that a good merge $(A,B)$ cannot invalidate another good merge.

Case two is that the two good merges are disjoint, and Ward does a merge that overlaps with both of them. Thus, we have two good merges $(A,B)$ and $(C,D)$, and Ward performs merge $(B,C)$. Since Ward does not do inner-cluster merges, either $B$ or $C$ is $2$-composed, w.l.o.g. say that $C$ is $2$-composed. If $B$ is $2$-composed as well, then $(A,B)$ and $(C,D)$ are in the same block, so they would not both be in $\mathcal{S}$. So $B$ is an inner cluster. If $C$ points to $B$, then $(B,C)$ is not lopsided and would have happened in phase $P5a$. Thus, $C$ points to $D$. 
If $A$ is an inner cluster, then $(B,C)$ is a T2 merge and would have happened in phase $P5c$. 
So say that $A$ is $2$-composed. $(A,B)$ is a good merge. It is not a T2 merge since $C$ is $2$-composed. It has to be a T3 merge, thus, $A$ points to $B$. Thus, $(B,C)$ would have happened in phase $P5d$: It is a lopsided merge with $B$ left of $C$, and the $2$-composed cluster left of $B$ points to $A$. 

We have seen that no merge of Ward can invalidate two merges from $\mathcal{S}$. 
Thus, we can now charge in the following way. The cost of the performed merge is bounded by the cost of any available merge. 
For each Ward step, we look whether it invalidates a merge from $\mathcal{S}$. If so, then we charge the performed merge to this good merge. If Ward's merge does not invalidate any merge from $\mathcal{S}$, we just arbitrarily charge a merge in $\mathcal{S}$ and mark it as invalid. In this manner, we can pay for $n - 2n_1 - n_r  - (k-n_1-1) - 1$ merges, i.e., we can pay until the number of clusters is reduced to $n - n_1 - n_r - (n - 2n_1 - n_r  - (k-n_1-1) - 1) = k$. That completes the proof of Theorem~\ref{thm:OneDimensional}.

\section{Conclusions}

We have initiated the theoretical study of the approximation guarantee of Ward's method. In particular, we have shown that Ward computes a 2-approximation on well-separated instances, which can be seen as the first theoretical explanation for its popularity in applications. We have also seen that its worst-case approximation guarantee increases exponentially with the dimension of the input and that it computes an~$\mathcal{O}(1)$-approximation on one-dimensional instances. 

These results leave room for further research. It would be particularly interesting to better understand the worst-case behavior of Ward's method. It is not clear, for example, if it computes a constant-factor approximation if the dimension is constant. Our analysis of the one-dimensional case is very complex and the factor hidden in the $\mathcal{O}$-notation is large. It would be interesting to simplify our analysis and to improve the approximation factor.

\bibliography{literature}
\bibliographystyle{plainurl}

\end{document}